\documentclass{article}

\usepackage{amsfonts,amsmath,amssymb,amsthm}
\usepackage{verbatim,float,url,enumerate}
\usepackage{graphicx,subfigure,epsfig,psfrag}
\usepackage{bm,dsfont,color,appendix}
\usepackage{adjustbox}
\usepackage{bbm}
\usepackage{natbib}
\usepackage{latexsym,graphicx}
\usepackage{setspace}
\usepackage{tikz}
\usetikzlibrary{arrows,positioning}
\usepackage[margin=.7 in]{geometry}
\usepackage{rotating}
\floatstyle{ruled}
\newfloat{algorithm}{tbhp}{loa}
\floatname{algorithm}{Algorithm}

\newtheorem{theorem}{Theorem}
\newtheorem{lemma}{Lemma}
\newtheorem{remark}{Remark}
 
\newtheorem{proposition}{Proposition}

\def\Err{{\rm Err}}
\def\E{{\rm E}}

\def\E{\mathbb{E}}

\def\Var{\mathrm{Var}}
\def\Cov{\mathrm{Cov}}

\def\ty{\tilde{y}}
\def\tbx{\tilde{\bold{x}}}
\def\tx{\tilde{x}}
\def\hbeta{\hat{\beta}}
\def\tbeta{\tilde{\beta}}

\def\bSigma{\bold{\Sigma}}

\def\hQ{\widehat{Q}}

\def\tQ{\widetilde{Q}}

\def\bx{{\bf x}}
 
\title {Post model-fitting exploration via  a ``Next-Door'' analysis}

\author{Leying  Guan \thanks{Dept. of  Statistics,
    Stanford Univ, leying.guan@gmail.com},  and  Robert Tibshirani \thanks{Depts. of Biomedical Data Sciences, and Statistics,
    Stanford Univ, tibs@stanford.edu}\\
Stanford University } 

\begin{document}

\maketitle

\begin{abstract}
We propose a simple method for evaluating the model that has been chosen by an adaptive regression procedure, our main focus being  the lasso.
This procedure deletes each chosen predictor and refits the lasso to get a set of  models that are ``close'' to the one chosen, referred to as ``base model".  If the deletion of a predictor leads to significant deterioration in the model's predictive power, the predictor is called indispensable; otherwise, the nearby model is called acceptable and can serve as a good alternative to the base model. This provides both an assessment of the predictive contribution of
each variable and a set of alternative models that may be used in place of the chosen model.  

In this paper, we will focus on the cross-validation (CV) setting and a model's predictive power is measured by its CV error, with base model  tuned by cross-validation.  We propose a method for comparing the error rates of the base model with that of nearby models, and a p-value for testing whether a predictor is dispensable. We also propose a new quantity called model score which works similarly as the p-value for the control of type I error. Our proposal is closely related to the LOCO (leave-one-covarate-out) methods of (\cite{rinaldo2016bootstrapping}) and less so, to Stability Selection (\cite{meinshausen2010stability}).

We call this procedure ``Next-Door analysis'' since it examines models close to the base model. It  can be applied to Gaussian regression data, generalized linear models, and  other supervised learning  problems with $\ell_1$ penalization. It could also be applied to best subset and stepwise regression procedures. We have  implemented it in the R language as a library to accompany the well-known {\tt glmnet} library.
\end{abstract}   
\section{Introduction}
We consider  the usual regression or classification situation: we have
samples $(\bx_i,y_i)$, $i=1,2,\ldots n$ where $\bx_i=(x_{i1},\ldots x_{ip})^T$ and $y_i$
are the regressors and response for the $ith$ observation. 
In regression, $y_i$ is quantitative while in classification it takes on one of $K$ discrete values. We will focus for now on the regression problem,
but will discuss classification in Section \ref{sec:extension}.

We assume that an adaptive regression procedure has been fit to the data, and we want to assess the chosen model(base model). Our main focus in this paper is on the lasso, although procedures such as subset or stepwise regression may also be amenable to our approach. The lasso method solves the following problem(for simplicity, we have left out the intercept):
 \begin{equation}
\hat\beta={\rm argmin} \frac{1}{2n} \sum_i (y_i-\bx^T_i\beta)^2+\lambda | \beta| 
 \end{equation}
yielding  a final model with sparse coefficients $\hat{\beta}$, for a sufficiently large value of $\lambda$. The data analyst is often interested in the importance of the selected predictors. 

One way to measure the importance is to adopt a sub-model interpretation where we consider whether a predictor has a  non-zero coefficient in the selected model. Conditional on the selected model, we can form post-selection p-values for the non-zero coefficients (\cite{berk2013valid, lee2014exact, lee2016exact, tibshirani2016exact, tibshirani2015uniform, fithian2014optimal}).   Another way to measure its importance is  to consider if the deletion of this predictor leads to significant deterioration in the predictive power given a training procedure. If the answer is ``yes", this predictor is indispensable. Otherwise,  the new model trained without this feature is acceptable and may work as a substitute for the base model.

The measures coincide when there is no feature selection. When $n >> p$,  we can fit a full regression model
\[
 \hat\beta = {\rm argmin} \frac{1}{2n} \sum_i (y_i-\bx_i^T\beta)^2.
\]
If we restrict ourselves to the OLS regression only, the p-value that we obtain for each predictor reflects both (1) the significance of its coefficient being non-zero conditional in the current model, and (2) the deterioration in predictive power when the predictor is deleted and the model is refitted.  When $p$ is large
and especially when $p>n$, a full regression fit is not feasible, and the lasso is a popular approach for fitting. When we use the lasso penalty to select a model, however, these two criteria are different. The recent progress in the field of post selection inference has focused on the sub-model interpretation. In practice, researchers will sometimes be more interested in the second perspective.

Motivated by the discussion above, we propose to find and assess  models ``close" to the base model, a procedure that we  call  {\em Next-Door analysis}.  The idea is as follows. We first fit the usual lasso, using cross-validation  to choose $\lambda$. Then for each predictor in the support set,  we remove that predictor and refit the lasso to all of the remaining  predictors (not just the support set) using the chosen value of $\lambda$. This gives a nearby  model(proximal model) corresponding to the deletion of
 each of the member of support set. Finally,  we examine and evaluate each of these nearby models.  
 Algorithm \ref{alg:SFS} gives the details. 
\medskip
\begin{algorithm}
\centerline{\bf Algorithm 1: Next-Door analysis for the   lasso}
\label{alg:SFS}
\begin{enumerate}
 \item Fit the lasso with parameter $\lambda$ chosen by cross-validation. Let the solution be $\hat\beta(\lambda)$.
 Let $S$ be the active set where the coefficient in $\hat\beta(\lambda)$ is non-zero.
 \item For each $j \in S$,  solve the lasso problem with the coefficient for the $j^{th}$ predictor being fixed at 0:
 \begin{equation}
\hat\beta(\lambda; j)={\rm argmin}_{\beta_j = 0} \frac{1}{2n}  \sum_i (y_i-\bx^T_i\beta)^2 +\lambda |\beta|
 \end{equation}
Let $d_{j}$ be the increase in the true validation error for this model relative to the base model. 
 \item  Form an unbiased estimate of $d_j$ and test if predictor $j$ is {\em indispensable}: that is, test whether  $d_j$  is positive.
\end{enumerate}
\end{algorithm}

As outlined in Algorithm \ref{alg:SFS}, our test for indispensibiltiy is a test of $H_0: d_j\geq 0$. It is challenging since the candidate models and the hypothesis are data adaptive and involve selections. One main task of this paper is to provide a good estimate of the p-value for the above test taking into consideration the selections.

 %In this paper, we consider the model selection and hypothesis selection separately and propose a flexible testing method that can be generalized to complicated settings. More specifically, we propose  a simple method based on the de-biased test error estimate and the bootstrap to take into account the fact of adaptive model fitting and we construct a model score to incorporate the hypothesis selection. 

\begin{table}[ht]
\centering
\caption{\em Prostate cancer results. The leftmost column shows the fitted model from the lasso, and the remaining columns show the nearby models
corresponding to the removal of each predictor.} 
\label{tab:prostate}.
\begin{tabular}{lrrrrrrrr}
  \hline
 & base & lcavol & lweight & svi & lcp & lbph & pgg45 & age \\ 
  \hline
lcavol & 0.64 &  & 0.69 & 0.70 & 0.59 & 0.65 & 0.63 & 0.62 \\ 
  lweight & 0.27 & 0.37 &  & 0.30 & 0.27 & 0.35 & 0.27 & 0.26 \\ 
  svi & 0.25 & 0.46 & 0.29 &  & 0.22 & 0.21 & 0.27 & 0.25 \\ 
  lcp & -0.12 & 0.07 & -0.11 & -0.01 &  & -0.14 & -0.04 & -0.11 \\ 
  lbph & 0.18 & 0.21 & 0.29 & 0.14 & 0.19 &  & 0.18 & 0.17 \\ 
  pgg45 & 0.17 & 0.18 & 0.13 & 0.19 & 0.13 & 0.18 &  & 0.15 \\ 
  age & -0.08 & -0.02 & -0.03 & -0.09 & -0.07 & -0.05 & -0.07 &  \\ 
  gleason &  & 0.07 &  &  &  &  & 0.07 &  \\ \hline
  cv\_error & 0.61 & 0.90 & 0.65 & 0.64 & 0.62 & 0.61 & 0.63 & 0.60 \\ 
  debiased\_error & 0.62 & 0.94 & 0.66 & 0.66 & 0.63 & 0.62 & 0.62 & 0.62 \\ 
    test\_error & 0.51 & 0.87 & 0.49 & 0.56 & 0.50 & 0.50 & 0.47 & 0.53 \\ 
   \hline
  selection frequency &  & 1.00 & 1.00 & 0.96 & 0.78 & 1.00 & 0.88 & 0.74 \\ 
  model pvalue &  & 0.01 & 0.21 & 0.20 & 0.29 & 0.48 & 0.26 & 0.34 \\ 
  model score &  & 0.01 & 0.21 & 0.21 & 0.37 & 0.48 & 0.30 & 0.45 \\ 
  feature pvalue &  & 0.00 & 0.01 & 0.02 & 0.23 & 0.05 & 0.07 & 0.28 \\ 

 \hline 
   \end{tabular}
\end{table}

Table \ref{tab:prostate}  gives a preview of results from a Next-Door analysis. We
apply it  to a prostate cancer data set taken from \cite{friedman2001elements}. The data consists of $n=67$ training observations and 30 test observations. There are eight predictors.
The response is the log PSA for men who had  prostate cancer surgery.  Each column contains one set of model coefficients using a fixed training procedure. The columns corresponding to the proximal models  are ordered according to their de-biased CV errors(from small to large).  Details of the model p-value and  model score for the  ``indispensability test'' are provided in  Section \ref{sec:test}. The ``selection frequency" is the proportion of times that the  predictor is selected when the model fitting procedure is applied 50 times to bootstrap samples. The ``feature p-value'' is a post-selection p-value testing for non-zero coefficients. It is obtained using the R package {\tt selectiveInference}(\cite{lee2016exact}) .    The feature p-values suggest that several predictors are significant, but only   {\tt lcavol}  is indispensable considering the out-of-sample performance according to the model p-value and model score.
For example, {\tt lweight} is highly significant according to the feature p-value but not by the other two measures.:  the test error results suggest that the coefficients  on other predictors can be adjusted to produce a model with no much worse out of sample performance.

\subsection{Related work}
\label{subsec:related}
Next-Door analysis  measures the importance of a predictor by whether we can find a good model excluding this feature. It is closely related to the LOCO parameters described in \cite{rinaldo2016bootstrapping} and the variable  importance measures used in random forest(\cite{breiman2001random}). In the work of \cite{rinaldo2016bootstrapping}, a hold-out data set is available. They do model selection, hypothesis selection and model fitting using only the training  data.  For each selected predictor, they coerce it to have a zero coefficient and rerun the model selection  and  training procedure. They then compare the performance of the original model and the new model  in a hold-out validation set to evaluate its importance. They are able to do model free inference conditional on the training data. Later, \cite{markovic2017adaptive} suggests the use of marginalized LOCO. In the procedure of training a model  with lasso penalty, instead of conditioning on the training data, they condition on a penalty being selected as well as the selected feature set $E$. They retrain the models with all data using OLS with features in $E$ and features in $E\setminus{j}$, and  compare instead the prediction errors of these two models after  marginalizing out the randomness in the training.  Next-Door analysis essentially looks at a different type of marginalized LOCO parameter, without restrict ourselves to the selected feature set $E$.  It is different from the work of \cite{rinaldo2016bootstrapping}  or \cite{markovic2017adaptive} in the following ways:
\begin{enumerate}
\item Next-door analysis considers a different marginalization level. We marginalize out all randomness including the parameter tuning.
\item We do not have a hold-out data set and we measure the importance of a feature by the test error of the CV models.
\item After the penalty is chosen with CV, we fix it when leaving out a predictor and retraining the model to loosely control the model complexity so that it is similar to the original model. We can also vary this $\lambda$ as in  \cite{rinaldo2016bootstrapping}.  However, it does not seem to be necessary when we have  marginalized out the randomness in the penalty picking step.
\end{enumerate}
The answer of which marginalization level to consider should depend on how people make prediction in practice. For example, if we do not retrain the model with new data coming in, the LOCO conditional on the training data in \cite{rinaldo2016bootstrapping}  is more proper. However, if we repeat the whole training procedure including the parameter tuning, we may want  a fully marginalized quantity. 

If we look at it from a different perspective, our proposal  is also related to the low dimensional projection estimator (LDPE)(\cite{zhang2014confidence, zhu2017breaking, yu2018confidence}).  These estimators are  concerned with the question of  whether a predictor is important conditioning on all other predictors.  To deal with the high dimensionality, LDPE is constructed using good initial model coefficient estimates and the part of a predictor that is ``almost" orthogonal to other predictors. Our approach deals with high dimensionality through a different perspective and restricts ourselves to a small set of ``accessible" models, which are models close to the base model in Next-Door analysis.  Instead of looking at the coefficients, it looks directly  at the prediction error. Another less related procedure is ``Stability selection'' (\cite{meinshausen2010stability}). This method  identifies a set of ``stable'' variables that are selected with probability above a threshold by procedures like the lasso.  Like the p-values from post-selection inference, even if a predictor is selected with  reasonably high probability, it is still possible that  we can  find an alternative among the reachable models with similar prediction performance.  For example, if we have two predictors that are identical and each of them is very important to the response without conditioning on the other, neither of them should be indispensable, but the selection probability will be around $0.5$ for each of them.

The paper is organized as follows. In Section \ref{sec:test}, we formalize how to test whether the difference in CV test errors $d_j$ is large with the full marginalization. We give details of the test method  and give the definition of the model score in this section.  Section \ref{sec:sim}, we provide intensive simulations to show the good performance of suggested methods.  We apply Next-Door analysis  to some real data examples in Section \ref{sec:realdata}.  In Section \ref{sec:extension}, we discuss the extension of Next-Door analysis to other settings.
  
 \section{Test for indispensability with full marginalization}
 \label{sec:test}
 In this section we give details of methods for the ``indispensability  test''   in Step (3) of Algorithm 1 above.  Let $\Lambda := \{\lambda_1, \ldots, \lambda_{m}\}$ be the set of penalty parameters that we consider and  suppose that we divide the data into $V$ folds $\cup^V_{v=1}\mathcal{V}_v$ with equal size.  For any fixed penalty $\lambda_k$, $k = 1,2,\ldots, m$,  let $S_k$ be the set of predictors selected.  The CV errors for models trained with and without predictor $j$ are $Q_k$ and $Q_k^j$, defined as
\[
Q_k = \frac{1}{n}\sum^n_{i=1}Q_k (\bx_i, y_i), \;\;Q^j_k=  \frac{1}{n}\sum^n_{i=1}Q^j_k(\bx_i, y_i)
\]
where $Q_k (x_i, y_i)$ and $Q^j_k(x_i, y_i)$ are the loss for the sample $(\bx_i, y_i)$ in CV:
\[
Q_k (\bx_i, y_i) = \sum^V_{v=1}(y_i-\bx^T_i\hat{\beta}^v(\lambda_k))^2\mathbbm{1}_{i\in \mathcal{V}_v},\;\;Q^j_k (\bx_i, y_i) = \sum^V_{v=1}(y_i-\bx^T_i\hat{\beta}^v(\lambda_k;j))^2\mathbbm{1}_{i\in \mathcal{V}_v}
\]
 where $\hat{\beta}^v(\lambda_k)$ and $\hat{\beta}^v(\lambda_k; j)$ are the coefficients trained using data $\cup_{v'\neq v}  \mathcal{V}_{v'}$ with penalty $\lambda_k$. Let ${\rm Err}_k$ and ${\rm Err}_k^j$ be the CV test error defined as the expectation of validation errors:
\[
\Err_k = E[Q_k], \;\;\Err^j_k= E[Q^j_k].
\]
In practice, we will pick $\lambda_k$ according to a criterion $R$. In this section, we consider the case where we pick $\lambda_k= \lambda_{k^*}$ to minimize the randomized validation error(we will discuss the randomized error later). Other criterion could also be used. For example, one can use the CV one standard error rule (\cite{friedman2001elements}). 

The index $k^*$  chosen  is a randomized quantity -- if we do the selection with a different random seed, we can  end up with a different penalty $\lambda_{k^*}$. As we do not want to make judgement about predictor $j$ based on a random quantity, we marginalize out the randomness in $k^*$ and end up with the marginalized test error under the criterion $R$. We let  $O_{k^*}$ to be the event of selecting penalty index $k^*$, and $(\bx, y)$ be an independent sample generated from their joint distribution. The test error after marginalization is defined as 
\[
\Err^{R} = \frac{1}{V}\sum^m_{k^* = 1}E[\sum^V_{v=1}(y-\bx^T\hat{\beta}^v(\lambda_{k^*}))^2\mathbbm{1}_{O_{k^*}}],\;\;\Err^{j, R} = \frac{1}{V}\sum^m_{k^* = 1}E[\sum^V_{v=1}(y - \bx^T\hat{\beta}^v(\lambda_{k^*};j))^2\mathbbm{1}_{O_{k^*}}]
\]
We are interested in the following hypothesis:

\[
H_0:  \Err^{j, R} \leq \Err^{R}   \;\;\;\;vs. \;\;\;\;H_1:\Err^{j, R}> \Err^{R} 
\]
%$\Err^{j, R} = \sum^m_{k^* = 1} \Err^j_{k^*} P(O_{k^*})$ and $\Err^{R} = \sum^m_{k^* = 1} \Err_{k^*} P(O_{k^*})$, and different testing problem:
%For any $j\in S_{k^*}$, we want to make a decision on whether it is dispensable and whether  $M_{k^*}^j$ is an acceptable proximal model to $M_{k^*}$. To do this, we ask if we can reject the following null hypothesis at a given level $\alpha$:
%\begin{align}
%&\;\;\;\;H_0: \Err^j_{k^*} \leq \Err_{k^*}  \nonumber\\
%vs. &\;\;\;\;\;\;\;\;\;\;\;\;\;\;\;\;\;\;\;\;\;\;\;\;\;\;\;\;\;\;\;\;\;\;\;\;\;\;\;\;\;\;\;\;\;\;\;\;\;\;\;\;(G_1)\nonumber\\
%&\;\;\;\;H_1:\Err^j_{k^*} > \Err_{k^*.}  \nonumber
%\end{align}
%We call this hypothesis  testing problem $G_1$. This evaluates the removal of feature $j$ conditional on the selected index $k^*$. In this paper, we will focus on the second testing problem.
%A naive approach will look at the validation errors $\{Q_{k^*}(\bx_i, y_i)\}$ and $\{Q_{k^*}(\bx_i, y_i)\}$, pretending $k^*$ and $j$ both to be fixed.  
The two events below prevent us from using the observed validation errors to do the test  directly :
\begin{enumerate}
\item {\bf Selection event} $A_1$(model selection):  The selected $\lambda_{k^*}$ penalty achieves the smallest randomized CV errors among all $\lambda_k\in \Lambda$.
\item  {\bf Selection event} $A_2$(hypothesis selection):  $j$ is in the non-zero support $S_{k^*}$.
\end{enumerate}
To make the proposed method more generalizable to complicated settings, we consider the event $A_1$ and $A_2$ separately. Intuitively, the event $\{j\in S_{k^*}\}$ should only have small effect:  the fact that the predictor $j$ is selected will not typically  have a big influence on the error of a refitted model  that excludes this predictor, when the number of covariates is moderately large. However, the validation error obtained after selection event $A_1$ can be significantly biased (\cite{tibshirani2009bias}).  

We give definition of the randomized cross-validation error and construct a de-biased test error estimate in Section \ref{subsec:random}. In Section  \ref{subsec:bootstrap}, we describe the Bootstrap p-value with the de-biased test error estimate considering only the event $A_1$. In Section \ref{subsec:score}, we  propose a new importance measure called the model score, which uses the previous p-value  to construct a quantity which can control the type I error after both selections $A_1$ and $A_2$.  From a practical view, we recommend the use of the model score if the cost of falsely rejecting the null hypothesis is high; otherwise, the Bootstrap p-value constructed in Section \ref{subsec:bootstrap} usually works well and has higher power when signal detection is hard.

\subsection{Randomized cross-validation error and the de-biased error estimate}
\label{subsec:random}
For simplicity of notation, for a pre-fixed predictor $j$, we let $Q = (Q_1, \ldots ,  Q_m,Q^j_{1}, \ldots ,  Q^j_{m})$ be the sequence of CV errors where the first $m$ are from models using all predictors and the next $m$ are from models with predictor $j$ left out. Let $\Err=  (\Err_1, \ldots ,  \Err_m,\Err^j_{1}, \ldots ,  \Err^j_{m})$ be the their underlying test errors. We define two sequences of randomized pseudo errors,
\begin{align}
& \tQ^{\alpha}(\epsilon, z)=Q+\frac{\epsilon}{\sqrt{n}}+\sqrt{\frac{\alpha}{n}}z,\;\;\tQ^{\frac{1}{\alpha}}(\epsilon, z) = Q+\frac{\epsilon}{\sqrt{n}}-\sqrt{\frac{1}{n\alpha}}z
\label{eqn:parallel}
\end{align}
where  $\epsilon \sim \mathcal{N}(0,\gamma_1\sigma^2_0 \bold{I})$, $z\sim N(0,\hat{\Sigma}+\gamma_1\sigma^2_0\bold{I})$ with $\gamma_1$ and $\alpha$ being a positive constant and  $\sigma^2_0$ being the smallest diagonal elements of $\hat{\Sigma}$, an estimate of the covariance of $\sqrt{n}Q$.

We choose the model index $k^*$ to minimize the randomized validation errors $\tQ^{\alpha}_k(\epsilon, z)$ for $k = 1,2,\ldots, m$.  In other words, we let the event  $O_{k^*} = \{\tQ^{\alpha}_{k^*}(\epsilon, z) \leq \tQ^{\alpha}_k(\epsilon, z), \forall k=1,\ldots, m\}$.

The first term $\frac{\epsilon}{\sqrt{n}}$ is proposed by \cite{rinaldo2016bootstrapping} the avoid the technical problem when applying CLT to the LOCO parameter  in the sample splitting case. It is also proposed in  \cite{markovic2017adaptive} to make the randomized CV curves asymptotically normal with invertible covariance structure under suitable assumptions, which have a similar style to  the consistency, range, moment and dimension assumptions below.
\begin{itemize}
\item Consistency assumption: For every $\lambda$ and predictor index $j$ considered,  the lasso estimator $\hat{\beta}(\lambda)$, $\hat{\beta}(\lambda;j)$ are consistent  to  some fixed vectors $\beta(\lambda)$  and $\beta(\lambda;j)$ at the rate $n^{\frac{1}{4}}$:
\begin{align*}
&n E\|\hat{\beta}(\lambda) -\beta(\lambda)\|^4_2 \rightarrow 0, \;\;n E\|\hat{\beta}(\lambda;j) -\beta(\lambda;j)\|^4_2 \rightarrow 0
\end{align*}
\item Range assumption: For any sample size $n$, we consider only the range of $\lambda$ such that $\lambda < Cn^{-\frac{1}{4}}$ for a large enough constant $C$.
  \item Moment assumption: 
 \begin{itemize}
\item $\rm{Var}((y - \bx^T\beta(\lambda))^2)$, \; $\rm{Var}((y- \bx^T\beta(\lambda;j))^2) $ are in the range $[c, C]$ for some positive constants  $c$ and $C$.
\item $E[\|\bx\|_2^2(y -\bx^T\beta(\lambda))^2] \leq C$, \;$E[\|\bx\|_2^2(y -\bx^T\beta(\lambda; j))^2] \leq C$ for some positive constant $C$.
\item $E[\|\bx\|_2^4] < \infty$
 \end{itemize}
  \item Dimension assumption: the dimension $p$ and the number of penalty parameters $m$   considered  is finite.
\end{itemize}
\begin{remark}
The range assumption indicates that the $\lambda$ we considered depends on the sample size $n$, which is also what happens in practice.  When there is non collinearity,  the $\lambda$ is considered to be sufficiently large if $\sqrt{n}\lambda \rightarrow \infty$ (\cite{wainwright2009sharp}). The range assumption is  very mild in this sense.
\end{remark}
The second terms $\sqrt{\frac{\alpha}{n}}z$ and $\sqrt{\frac{1}{n\alpha}}z$ are introduced to make $\tQ^{\alpha}$ and $\tQ^{\frac{1}{\alpha}}$ marginally  and asymptotically independent under the assumptions above. This kind of parallel construction is proposed in \cite{harris2016prediction}. In their work, the author estimates the prediction error for estimators like relaxed LASSO in the linear regression when the noise in the response $y$ is homoscedastic Gaussian with variance $\sigma^2$. They also create two marginally independent responses $y^{\alpha}$ and $y^{\frac{1}{\alpha}}$ by adding noises $\sqrt{\alpha}\epsilon$ and $\frac{\epsilon}{\sqrt{\alpha}}$  to $y$ with $\epsilon\sim N(0, \sigma^2)$. Marginally, the prediction error estimate with $y^{\frac{1}{\alpha}}$ is unbiased for any selection performed using $y^{\alpha}$. When $\sqrt{n}Q$  is asymptotically normal, we also get an almost unbiased test error estimate using $\tQ^{\frac{1}{\alpha}}$ after selecting the model using $\tQ^{\alpha}$(\cite{guan2018test}). Algorithm \ref{alg:debias} gives details of the de-biased test error estimate and Theorem \ref{thm:randomized0} states that this procedure can successfully reduce the bias under assumptions above.
\begin{algorithm}
\centerline{\bf Algorithm \ref{alg:debias}: Debias Error Estimation with Randomization}
\label{alg:debias}
\begin{enumerate}
\item  Input the $n\times 2m$ error matrix $Q(\bx_i, y_i)$ and parameters $\alpha$, $\sigma^2_0$ , the number of repetitions $H$ and covariance $\hat{\Sigma}$. By default, we set $\gamma_1=\alpha= 0.1$, $H = 1000$. The default for  $\hat\Sigma$  is the sample covariance matrix. 
\item Generate  Gaussian noise $(\epsilon, z)$ and let $k^*$ be the index chosen using $ \widetilde{\rm{Q}}^{\alpha}(\epsilon, z)$.
\item Generate  $H$ samples of the additive noise pair: at the $h^{th}$ round, let  $(\epsilon_h, \; z_{h})$ be the random vector generated and $k^*_h$  be the index chosen.  The de-biased errors are given by:
$$\widehat{\Err} =\frac{1}{H}\sum^H_{h=1} \tQ^{\frac{1}{\alpha}}_{k^*_h}(\epsilon_h, z_h), \;\;\widehat{\Err}^j =\frac{1}{H}\sum^H_{h=1} \tQ^{\frac{1}{\alpha}}_{m+k^*_h}(\epsilon_h, z_h).$$
\item Output the de-biased error estimates: $\widehat{\Err}$, $\widehat{\Err}^j $.
\end{enumerate}
\end{algorithm}

Let $\Sigma$ be the covariance structure of $((y-\bx^T\beta(\lambda_1))^2, \ldots, (y-\bx^T\beta(\lambda_m))^2, (y-\bx^T\beta(\lambda_1;j))^2, \ldots, (y-\bx^T\beta(\lambda_m;j))^2)$.
\begin{theorem}
\label{thm:randomized0}
Suppose the consistency, range, moment and dimension assumptions hold. Let $\hat{\Sigma} $ be an estimate of $\Sigma$. If this estimate satisfies the following two requirements (1) $\|\hat{\Sigma} - \Sigma\|_{\infty}\overset{p}{\rightarrow} 0$ and (2) $E[\sum^{2m}_{k=1} \hat{\Sigma}_{k,k}] \leq C$ for some constant $C$, then we have
\begin{align*}
&\sqrt{n}(E[ \widehat{\Err}] - \Err^{R}) \rightarrow 0,\;\;\sqrt{n}(E[\widehat{\Err}^j]-\Err^{j, R}) \rightarrow 0
\end{align*}
\end{theorem}

\begin{lemma}
\label{lem:covariance}
Suppose that the consistency, range, moment and dimension assumptions hold. Let $\hat{\Sigma} $ be the sample covariance structure:
\[
\hat{\Sigma}_{k, k' } = \frac{\sum^n_{i=1} (Q_k(\bx_i, y_i) - Q_k)(Q_{k'}(\bx_i, y_i) - Q_{k'})}{n}
\]
We have (1) $\|\hat{\Sigma} - \Sigma\|_{\infty}\overset{p}{\rightarrow} 0$ and (2) $E[\sum^{2m}_{k=1} \hat{\Sigma}_{k,k}] \leq C$ for some constant $C$.
\end{lemma}
\begin{remark}
Instead of estimating $\Sigma$ using paired Bootstrap which requires huge computational cost for large $p$, Lemma \ref{lem:covariance} suggests that we can use the sample covariance matrix estimate under assumptions in this paper. Another heuristic way to justify  the use of sample covariance estimate is from a perspective conditional on the training model, which can be found in \cite{lei2017cross}.
\end{remark}
Under assumptions above, we can also write $\Err^R$ and $\Err^{j, R}$ are weighted test error in terms of $\Err$. 
\begin{lemma}
\label{lem:expectation}
Suppose that the consistency, range, moment and dimension assumptions hold, then we have
\begin{align*}
&\sqrt{n}\left(\Err^R - \sum^m_{k^*=1}\Err_{k^*} P(O_{k^*}) )\right)\rightarrow 0,\;\;\sqrt{n}\left(\Err^{j,R} - \sum^m_{k^*=1}\Err_{k^*+m} P(O_{k^*}) \right)\rightarrow 0
\end{align*}
\end{lemma}

Proofs of Theorem \ref{thm:randomized0} and Lemma  \ref{lem:covariance},  \ref{lem:expectation} are in Appendix \ref{app:proof}.
\subsection{Bootstrap p-value approximation}
\label{subsec:bootstrap}
We  look at the quantity $T=(\widehat{\Err}^j -\widehat{\Err} ) -(\rm{Err}^{j, R} -\rm{Err}^{R})$.  Theorem \ref{thm:randomized0} states that no matter what the parameter for the underlying population is, the $\sqrt{n}T$ will have mean very close to 0.  Hence we bootstrap this test statistic to approximate the true p-value. We expect that the cumulative distribution function of $T^*$, the test statistics from the Bootstrap sample, will be close to that of $T$: let $F(x)$ be the CDF of $T$ and $F^*(x)$ be the CDF of $T^*$, we take the approximation that $F(x) \approx F^*(x)$. The p-value for the null hypothesis is then calculated as  $p = 1-F^*(\widehat{\Err}^j -\widehat{\Err} )$.  Some  corrections can be introduced to improve the empirical performance of the type I error control.  Here, we apply two modifications:
\begin{enumerate}
\item It is possible that the distributions of $\sqrt{n}T^*$ and  $\sqrt{n}T$ are asymptotically degenerate. To account for this case,  instead of looking at the empirical distribution of $T^*$, we look at the empirical distribution of $T^*+w$, where $w\sim \frac{\gamma_2}{\sqrt{n}}N(0, \sigma^2_0)$ for a small constant $\gamma_2$. The larger $\gamma_2$ is, the less power we will have and we will be more conservative.
\item Let $\zeta_k = \sqrt{n}(Q_k - \Err_k)$, and  $\overline{\Err}_1$, $\overline{Q}_1$, $\overline{\zeta}_1$ be the means of the test error, validation error and $\zeta_k$ for the first $m$ models. We know that
\begin{align*}
&E[\sum^m_{k=1}(Q_k - \overline{Q}_1)^2] = E[\sum^m_{k=1}(\Err_k - \overline{\Err}_1+\frac{\zeta_k}{\sqrt{n}} - \frac{\overline{\zeta}_1}{\sqrt{n}})^2] = E[\sum^m_{k=1}(\Err_k - \overline{\Err}_1)^2]+\frac{\sum^m_{k=1}\Sigma_{k,k}}{n} - \frac{\sum^m_{k, k'=1}\Sigma_{k,k'}}{nm}
\end{align*}
We see that the Bootstrap population have inflated underlying test error dispersion due to noise.  Let $Q^s$ be a vector of size $2m$. To match the average variability among the first $m$ models'  test errors, we can let
\[
Q^s_{k} = \sqrt{\frac{[\sum^m_{k=1}(Q_k - \overline{Q}_1)^2 - (\frac{\sum^m_{k=1} \Sigma_{k,k}}{n} - \frac{\sum^m_{k,k' = 1}\Sigma_{k,k'}}{n m })]_+}{\sum^m_{k=1}(Q_k - \overline{Q}_1)^2 }}(Q_k-\overline{Q}_1)+\overline{Q}_1\,\; \forall k = 1,2,\ldots, m
\]
Similarly, let $\overline{Q}_2$ be the mean validation error  for the $m$ models in the second half, to match the average variability among the second $m$ models'  test errors, we let
\[
Q^s_{k} = \sqrt{\frac{[\sum^{2m}_{k={m+1}}(Q_k - \overline{Q}_2)^2 - (\frac{\sum^{2m}_{k=m+1} \Sigma_{k,k}}{n} - \frac{\sum^{2m}_{k,k' = m+1}\Sigma_{k,k'}}{n m })]_+}{\sum^{2m}_{k={m+1}}(Q_k - \overline{Q}_2)^2 }}(Q_k-\overline{Q}_2)+\overline{Q}_2\,\; \forall  k = m+1,m+2,\ldots, 2m
\]
The mean-rescaled Bootstrap is to do dootstrap in the population with the population mean $Q^s$ instead of $Q$. Let $Q(\bx_i, y_i)$ be $2m$ vector representing the the loss for sample $i$. The mean-rescaled bootstrap generates samples from the mean-rescaled population:
\begin{equation}
\label{eq:rescaleDist}
Q^*(\bx_i, y_i) \overset{i.i.d}{\sim} \{Q(\bx_1, y_1), Q(\bx_2, y_2), \ldots,  Q(\bx_n, y_n)\} - Q+Q^s
\end{equation}
The mean-rescaled Bootstrap statistics $T^*$ is the realization of the test statistics from the distribution above.  The p-value testing $H_0$ is  constructed as $p = P(\widehat{\Err}^{j} - \widehat{\Err} \leq T^*+w)$. We reject the null hypothesis when $p \leq \alpha$.
\end{enumerate}

We provide simulations of this approximate p-value's distribution in Appendix \ref{app:bootstrap},  the results show that approximate p-value using Bootstrap with the de-biased estimates is more uniform compared  to that from the bootstrap using the observed CV errors.

\noindent\rule{\textwidth}{0.4pt}\\
\centerline{\bf Bootstrap p-value approximation}
\begin{enumerate}
 \item Input: Level $\alpha$, the errors $\{Q(\bx_i, y_i), \; i = 1,2,\ldots, n\} $ (first $m$ correspond to the original model and the last $m$ correspond to models excluding predictor $j$), the number of bootstrap repetitions $B$, the extra noise level $\gamma_2$. By default, $B = 10000$ and $\gamma_2 = 0.05$.
 
 \item  De-biased estimate: We apply  Algorithm \ref{alg:debias} to get the de-biased test error estimate for the $\widehat{\Err}$ and $\widehat{\Err}^j$.
 \item Let $\{Q^s(\bx_i, y_i), \; i = 1,2,\ldots, n\}$ be the rescaled Bootstrap populations  defined in the right-hand-side of eq.(\ref{eq:rescaleDist}).
 \item Bootstrap: for each iteration $b$, we draw bootstrap samples from $\{Q^s(\bx_i, y_i), \; i = 1,2,\ldots, n\}$ and apply  Algorithm  \ref{alg:debias} to the Bootstrap samples. Let  $\widehat{\Err}_b$ and $\widehat{\Err}^j_b$ be the de-biased error estimates and $w_b\sim N(0, \frac{\gamma_2^2\sigma^2_0}{n})$, then the p-value is calculated as 
  \[
  \hat{p} = \frac{\sum^B_{b=1} \mathbbm{1}_{\{ \widehat{\Err}^j_b - \widehat{\Err}_b -Q^{s, R}- Q^{s, j, R}+w_b\geq \widehat{\Err}^j-\widehat{\Err} \} }}{B}
  \]
  where  $Q^{s, R}$ and $Q^{s,j, R}$ are the test errors under criterion $R$ using the mean-rescaled Bootstrap population:
  \[
  Q^{s, R} = E[\sum^m_{k^*=1}Q^s_{k^*}\mathbbm{1}_{O_{k^*}}],\;\;Q^{s, j, R} = E[\sum^m_{k^*=1}Q^s_{k^*+m}\mathbbm{1}_{O_{k^*}}]
  \]
  
   We reject the null hypothesis for the predictor $j$ if $\hat{p} \leq \alpha$.
 \end{enumerate}
\noindent\rule{\textwidth}{0.4pt}

\subsection{Model score: a conservative measure of importance}
\label{subsec:score}

In this section, we consider an additional post-processing step  to deal with the selection event $A_{2}$ and to guard against being overly optimistic. It first ignores $A_2$ and then accounts for it by discounting the importance of a predictor based on how frequently it is selected by the model. Let $p$ be a p-value considering only the selection event $A_1$, $O_{k^*}$ being the event such that we select the penalty $\lambda_{k^*}$ with criterion $R$. The model score is defined as $s = \frac{p}{\gamma_j}$,  where   $\gamma_j :=  \sum^m_{k^*=1}P(j\in S_{k^*}, O_{k^*})$ is the average selection frequency for predictor $j$ with criterion $R$. 
\begin{itemize}
\item Selection frequency assumption: as $n\rightarrow\infty$, the selection frequency of a predictor $j$  converges to a constant $\lim_{n\rightarrow\infty}\gamma_j = c_j\in[0,1]$
\end{itemize}
We only consider predictors whose selection frequency has limit greater than 0. For those predictors with non-vanishing selection frequency, we can control the type I error at level $\alpha$ asymptotically by rejecting only $s_j < \alpha$(Theorem \ref{thm:conservative}).

\begin{theorem}
\label{thm:conservative}
Let predictor $j$ be a predictor with $c_j > 0$, let $p$ be a the p-value constructed and $s = \frac{p}{\gamma_j}$. If the $p$ satisfies $\lim_{n\rightarrow \infty}P_{H_0}(p \leq \alpha ) \leq \alpha$ for any fixed $\alpha$, then we have  $\lim_{n\rightarrow \infty}P_{H_0}(s \leq \alpha|j\in S_{k^*}) \leq \alpha$.
\end{theorem}
\begin{proof}
By definition, $P(j\in S_{k^*}) = \gamma_j$. Let $\frac{0}{0}:= \infty$, then $P_{H_0}(s \leq \alpha|j\in S_{k^*}) =\frac{P_{H_0}(p\leq \gamma_j\alpha, j\in S_{k^*})}{\gamma_j} \leq \frac{P_{H_0}(p\leq \gamma_j\alpha)}{\gamma_j}$.   We take the limit of the above inequality and apply the Slutsky's theorem to conclude the proof: $\lim_{n\rightarrow\infty} \frac{P_{H_0}(p\leq \gamma_j\alpha)}{\gamma_j} =\lim_{n\rightarrow\infty}\frac{P_{H_0}(p\leq c_j\alpha)}{c_j} \leq \alpha$.
\end{proof}
The denominator of the model score $s$ is $\gamma_j$, the frequency of a predictor being selected using criterion $\mathcal{R}$, which is also used by stability selection (\cite{meinshausen2010stability}).  In practice, we can estimate $\gamma_j$ by doing a paired bootstrap of $(X, Y)$, refitting the models, and picking the $k^*$ using the new error. We then estimate the frequency of predictor $j$ being selected in those models.  Also, we will set a small cut-off, say, 0.05, on the observed selection frequency $\gamma_j$ and we do not reject a predictor if $\gamma_j$ is smaller than that.

\section{Simulations}
\label{sec:sim}
In this section, we evaluate the  performance of our proposal in the  linear regression setting and compare them to predictor p-values from post-selection inference and a naive approach neglecting all selections when looking at the model errors.  We consider both the de-biased approach neglecting selection event $A_2$ ``{\em model pvalue}",  and the de-biased approach using model score to account for $A_2$ ``{\em model score}". The p-values using the naive approach is referred to as ``{\em model pvalue(naive)}". For the post selection inference approach, we consider the post selection feature p-value  ``{\em feature pvalue}" from the  {\tt selectiveInference} package(\cite{lee2016exact}), the post selection model p-value``{\em model pvalue(post selection)}" neglecting the selection event $A_2$, as described in Appendix \ref{sec:postSelection}. We include the later to support our claim that the event $A_2$ does not have significant influence.  At any given level $\alpha$,  a rejection using the model score is the same as a rejection using a p-value: we reject a hypothesis if its score is smaller than a given level $\alpha$. 

We generate $n=100$ observations  from a linear model
$$y_i=\beta_0+\sum_{j=1}^pX_{ij}\beta_j +Z_i$$
for different dimensions $p$ and sparsity levels $ s$. When $s \neq 0$, we set $\beta=(\frac{2}{1},\frac{2}{2},\ldots, \frac{2}{s}, 0,\ldots 0)$. For a given $(s, p)$, we examine the following four simulation settings:

\medskip

\textbf{Orthogonal Design}:  Let $X$ be standard Gaussian predictors and $Z_i$ be standard Gaussian.

\medskip

\textbf{Redundant Design I}: This is a setting  designed specifically for Next-Door analysis. The design matrix is in a way such that almost no predictor is indispensable. Let the first half predictors $X_{1:\frac{p}{2}}$ and $W$ be standard Gaussian predictors with length $\frac{p}{2}$, and the second half predictors $X_{(\frac{p}{2}+1): p} = 0.95X_{(1:\frac{p}{2})} + 0.05W$, and $Z_i$ standard Gaussian.

\medskip

\textbf{Correlated Design}: Let $X$ be Gaussian predictors with variance $1$ and ${\rm corr}(X_j,X_k)=0.5$, and $Z_i$ be standard Gaussian.

\medskip

\textbf{Redundant Design II}: Let the first half predictors $X_{1:\frac{p}{2}}$ be Gaussian with variance 1 and ${\rm corr}(X_j,X_k)=0.5$, and $W$ be standard Gaussian predictors with length $\frac{p}{2}$. The second half predictors $X_{(\frac{p}{2}+1): p} = 0.95X_{1:\frac{p}{2}} + 0.05W$, and $Z_i$ standard Gaussian.

Empirical type I error results for targeted $90\%$ coverage are given in  Table \ref{tab:typeIerrorsp}. For the type I error calculation, we look at the settings with $s \leq \frac{p}{2}$: We consider predictors with index $j\geq s+1$ in the non-redundant case and $j = \frac{p}{2}+1, \ldots, \frac{p}{2}+5$ in the redundant case. The power curves under non-redundant design with $(p, s) = (10, 5), (10, 10), (400, 5), (400, 30)$ are given in Figure \ref{fig:res1}.

\begin{table}[ht]
\centering
\caption{\em Empirical Type I error with pre-specified level $\alpha = 0.1$ under four different experiment settings and different $(p, s)$ combinations.}
\label{tab:typeIerrorsp}
\begin{adjustbox}{width = \textwidth, height = .07\textwidth}
\begin{tabular}{|l|rrrr|rrrr|rrrr|}
  \hline
 & &(p, s) = (10, 0) &&&& (p,s) = (10, 5) &&&&(p,s)=(10, 10)&& \\ 
  \hline
 & Orthogonal & RedundantI & Correlated & RedundantII & Orthogonal & RedundantI & Correlated & RedundantII & Orthogonal & RedundantI & Correlated & RedundantII \\ 
  \hline
model pvalue(naive) & 0.05 & 0.12 & 0.17 & 0.19 & 0.04 & 0.30 & 0.04 & 0.17 &  &  &  &  \\ 
  model pvalue & 0.03 & 0.06 & 0.08 & 0.05 & 0.03 & 0.10 & 0.04 & 0.06 &  &  &  &  \\ 
  model score & 0.02 & 0.04 & 0.06 & 0.03 & 0.02 & 0.09 & 0.04 & 0.05 &  &  &  &  \\ 
  model pvalue(post selection) & 0.07 & 0.08 & 0.08 & 0.08 & 0.07 & 0.08 & 0.08 & 0.08 &  &  &  &  \\ 
  feature pvalue & 0.14 & 0.16 & 0.13 & 0.21 & 0.13 & 0.54 & 0.09 & 0.56 &  &  &  &  \\ 
\end{tabular}
\end{adjustbox}
\begin{adjustbox}{width = \textwidth, height = .07\textwidth}
\begin{tabular}{|l|rrrr|rrrr|rrrr|}
  \hline
 & &(p, s) = (400, 0) &&&& (p,s) = (400, 5) &&&&(p,s)=(400, 30)&& \\ 
  \hline
 & Orthogonal & RedundantI & Correlated & RedundantII & Orthogonal & RedundantI & Correlated & RedundantII & Orthogonal & RedundantI & Correlated& RedundantII \\ 
  \hline
model pvalue(naive) & 0.05 & 0.13 & 0.09 & 0.10 & 0.09 & 0.09 & 0.15 & 0.05 & 0.38 & 0.51 & 0.20 & 0.34 \\ 
  model pvalue & 0.13 & 0.09 & 0.12 & 0.11 & 0.13 & 0.06 & 0.12 & 0.08 & 0.13 & 0.13 & 0.12 & 0.11 \\ 
  model score & 0.07 & 0.05 & 0.06 & 0.05 & 0.06 & 0.05 & 0.07 & 0.06 & 0.06 & 0.11 & 0.05 & 0.09 \\ 
  model pvalue(post selection) & 0.10 & 0.10 & 0.11 & 0.10 & 0.09 & 0.11 & 0.11 & 0.13 & 0.11 & 0.12 & 0.09 & 0.12 \\ 
  feature pvalue & 0.10 & 0.16 & 0.15 & 0.15 & 0.15 & 0.13 & 0.14 & 0.14 & 0.14 & 0.35 & 0.16 & 0.10 \\ \hline
\end{tabular}
\end{adjustbox}
\end{table}

From our simulations, both the bootstrap model p-value and post selection p-value considering only $A_1$ have reasonable performance in controlling the type I error on average. The model scores are conservative and perform well in controlling the type I error. The naive approach and feature p-value can not control the type I error as expected. When we look at Figure \ref{fig:res1}, we see that
\begin{enumerate}
\item Comparing the bootstrap model p-value for the marginalized test error and the post selection model p-value conditional on the penalty selected,  we can see that there is loss in power for the latter.
\item Comparing  model p-values with the feature p-value, we can see that conditioning on the whole selected feature set $E$ can lead to dramatic power loss in high dimensional setting.
\end{enumerate}
In practice, the model p-value that neglects the selection $A_1$ is generally well-behaved.  It also has less computational cost and higher power in extremely small signals. However, the  model score approach may be preferred in cases where exact  type-I error  control is essential.
 \begin{figure}
\begin{center}
\includegraphics[width=1\textwidth, height =.6\textwidth]{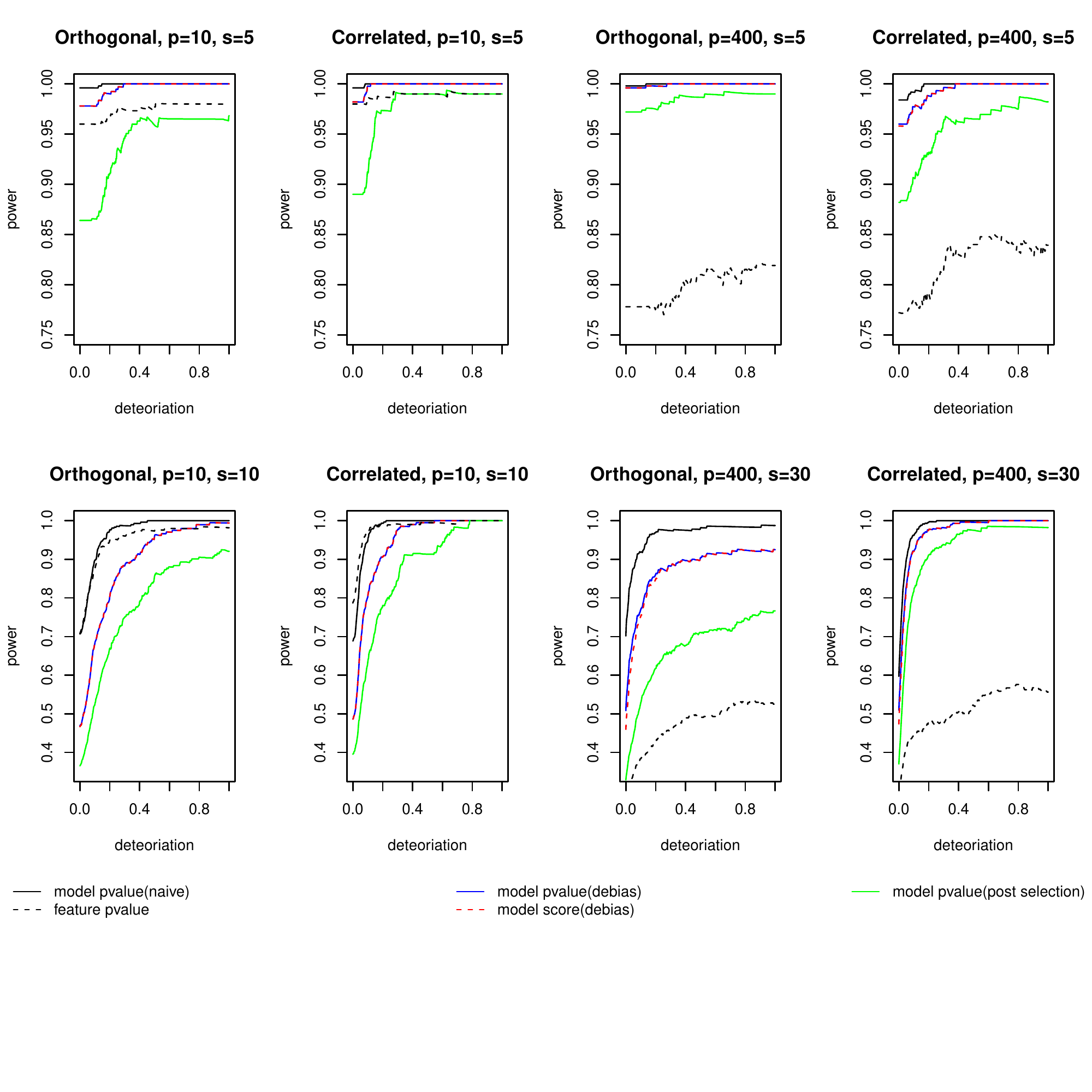}
\vskip -.5in
\caption{\em Power curves for six approaches with $(p, s) = (10, 5), (10, 10), (400, 5), (400, 30)$. The solid  blue and dashed red   curves are the power curve using the bootstrap p-value and its model score. The solid green  curve is that uses the post selection model p-value without considering $A_2$. The  solid black curve uses the t-test and nominal CV error  that ignores all selection events. The dashed black curve uses post selection for features.}
\label{fig:res1}
\end{center}
\end{figure}
\section{Real data applications}
\label{sec:realdata}
In this section, we provide  two more real data examples.  The second example uses the HIV data (\cite{rhee2003human}) where the author studied six nucleoside reverse transcriptase inhibitors that are used to treat HIV-1. We take the measurement of one of the inhibitors 3TC as the response, and the predictors are 240 mutation sites. There are 1073 samples in this experiment. We randomly split the samples into 800 training samples and 273 test samples.  The mutation site p184 is special, and has prediction power dominating all other sites. Results are given in Table \ref{tab:HIV}. The original randomized model selected 21 predictors. For the sake of space,  we do not include in the table those predictors whose model p-value and feature p-value are both greater than 0.2, and the test error for its proximal model is no greater than the test error for the original model(15 left). With a p-value cut-off being 0.1, four predictors, p184, p65, p215 and p69,  are found indispensable. For the rows, we do not show predictors that only appear in the models deleting p184, p65, p215 or p69. 

\begin{table}[h]
\centering
\caption{\em HIV dataset results. The leftmost column shows the fitted model lasso, and the remaining columns show the proximal model
corresponding to the removal of each predictor.} 
\label{tab:HIV}
\begin{adjustbox}{width = \textwidth}
\begin{tabular}{lrrrrrrrrrrrrrrrr}
  \hline
 &base & p184 & p65 & p215 & p69 & p228 & p33 & p172 & p75 & p54 & p210 & p67 & p115 & p90 & p151 & p62 \\ 
  \hline
p184 & 0.934 &  & 0.933 & 0.936 & 0.934 & 0.934 & 0.934 & 0.934 & 0.933 & 0.934 & 0.933 & 0.934 & 0.934 & 0.934 & 0.934 & 0.934 \\ 
  p65 & 0.110 & 0.083 &  & 0.105 & 0.110 & 0.110 & 0.110 & 0.110 & 0.110 & 0.110 & 0.111 & 0.110 & 0.111 & 0.110 & 0.111 & 0.112 \\ 
  p215 & 0.082 & 0.162 & 0.063 &  & 0.083 & 0.083 & 0.082 & 0.082 & 0.081 & 0.081 & 0.084 & 0.093 & 0.082 & 0.082 & 0.081 & 0.083 \\ 
  p69 & 0.024 & 0.023 & 0.021 & 0.025 &  & 0.026 & 0.023 & 0.024 & 0.024 & 0.024 & 0.023 & 0.028 & 0.024 & 0.023 & 0.024 & 0.024 \\ 
  p228 & 0.011 & 0.013 & 0.006 & 0.015 & 0.017 &  & 0.011 & 0.011 & 0.012 & 0.011 & 0.012 & 0.005 & 0.011 & 0.012 & 0.011 & 0.012 \\ 
  p33 & 0.004 & 0.010 & 0.003 & 0.005 & 0.003 & 0.004 &  & 0.004 & 0.005 & 0.004 & 0.004 & 0.003 & 0.004 & 0.004 & 0.004 & 0.003 \\ 
  p172 & 0.002 & -0.034 & 0.001 & 0.001 & 0.004 & 0.003 & 0.002 &  & 0.002 & 0.002 & 0.003 &  & 0.002 & 0.003 & 0.002 & 0.002 \\ 
  p75 & 0.009 &  & 0.004 & 0.004 & 0.010 & 0.010 & 0.009 & 0.009 &  & 0.008 & 0.010 & 0.013 & 0.009 & 0.010 & 0.009 & 0.012 \\ 
  p54 & -0.011 &  & -0.010 & -0.010 & -0.012 & -0.011 & -0.011 & -0.011 & -0.010 &  & -0.013 & -0.009 & -0.011 & -0.012 & -0.011 & -0.012 \\ 
  p210 & 0.015 & 0.019 & 0.019 & 0.025 & 0.014 & 0.017 & 0.015 & 0.016 & 0.016 & 0.016 &  & 0.018 & 0.015 & 0.015 & 0.015 & 0.015 \\ 
  p67 & 0.038 & 0.041 & 0.039 & 0.056 & 0.043 & 0.035 & 0.038 & 0.038 & 0.039 & 0.037 & 0.039 &  & 0.038 & 0.038 & 0.038 & 0.037 \\ 
  p116 & 0.002 & 0.005 &  & 0.003 & 0.003 & 0.002 & 0.002 & 0.002 & 0.001 & 0.002 & 0.002 & 0.002 & 0.002 & 0.002 & 0.009 & 0.002 \\ 
  p115 & 0.000 & 0.085 & 0.009 &  &  & 0.001 & 0.001 & 0.000 & 0.001 & 0.001 & 0.000 & 0.000 &  &  & 0.000 & 0.001 \\ 
  p90 & 0.009 & 0.044 & 0.009 & 0.010 & 0.008 & 0.010 & 0.009 & 0.009 & 0.010 & 0.010 & 0.009 & 0.009 & 0.009 &  & 0.009 & 0.009 \\ 
  p118 & 0.010 &  & 0.012 & 0.009 & 0.011 & 0.011 & 0.010 & 0.010 & 0.011 & 0.010 & 0.013 & 0.011 & 0.010 & 0.009 & 0.010 & 0.009 \\ 
  p77 & 0.006 & 0.005 & 0.018 & 0.005 & 0.005 & 0.006 & 0.006 & 0.006 & 0.009 & 0.007 & 0.007 & 0.007 & 0.006 & 0.006 & 0.008 & 0.007 \\ 
  p151 & 0.010 &  & 0.013 & 0.008 & 0.010 & 0.010 & 0.010 & 0.010 & 0.010 & 0.010 & 0.009 & 0.011 & 0.010 & 0.010 &  & 0.010 \\ 
  p62 & 0.010 & 0.051 & 0.025 & 0.014 & 0.011 & 0.010 & 0.009 & 0.010 & 0.012 & 0.010 & 0.009 & 0.007 & 0.010 & 0.010 & 0.010 &  \\ 
  p181 & 0.001 & -0.056 & 0.014 & 0.003 & 0.002 & 0.002 & 0.001 & 0.001 & 0.001 & 0.001 & 0.001 &  & 0.001 & 0.002 & 0.001 & 0.001 \\ 
  p41 & 0.007 & 0.186 & 0.003 & 0.053 & 0.006 & 0.008 & 0.007 & 0.007 & 0.008 & 0.007 & 0.014 & 0.005 & 0.007 & 0.007 & 0.007 & 0.008 \\ 
  p219 & 0.007 & 0.035 & 0.004 & 0.013 & 0.009 & 0.011 & 0.007 & 0.007 & 0.006 & 0.008 & 0.005 & 0.028 & 0.007 & 0.007 & 0.007 & 0.007 \\ 
  p25 &  &  &  &  &  &  & 0.001 &  &  &  &  &  &  &  &  &  \\ 
  p125 &  &  & 0.001 & 0.002 & 0.001 & 0.000 & 0.000 & 0.000 & 0.000 & 0.000 &  &  &  & 0.000 &  &  \\ 
  p200 &  & 0.006 & -0.003 &  &  &  &  & 0.000 &  &  &  & -0.001 &  &  &  & -0.001 \\ \hline
  cv\_error & 0.062 & 0.828 & 0.078 & 0.064 & 0.063 & 0.062 & 0.062 & 0.062 & 0.062 & 0.062 & 0.062 & 0.063 & 0.062 & 0.062 & 0.062 & 0.062 \\ 
  debiased\_error & 0.063 & 0.847 & 0.078 & 0.065 & 0.064 & 0.064 & 0.063 & 0.063 & 0.063 & 0.063 & 0.063 & 0.063 & 0.063 & 0.063 & 0.062 & 0.062 \\ 
    test\_error & 0.063 & 0.872 & 0.085 & 0.065 & 0.064 & 0.063 & 0.063 & 0.063 & 0.063 & 0.061 & 0.063 & 0.063 & 0.063 & 0.063 & 0.063 & 0.063 \\ 
   \hline
  selection frequency &  & 1.000 & 1.000 & 1.000 & 1.000 & 0.900 & 0.660 & 0.840 & 0.740 & 0.680 & 1.000 & 1.000 & 0.800 & 1.000 & 0.780 & 0.760 \\ 
  model pvalue &  & 0.000 & 0.000 & 0.000 & 0.037 & 0.689 & 0.152 & 0.503 & 0.505 & 0.231 & 0.690 & 0.141 & 0.838 & 0.433 & 0.463 & 0.378 \\ 
  model score &  & 0.000 & 0.000 & 0.000 & 0.037 & 0.765 & 0.231 & 0.599 & 0.682 & 0.339 & 0.690 & 0.141 & 1.047 & 0.433 & 0.593 & 0.497 \\ 
  feature pvalue &  & 0.000 & 0.000 & 0.007 & 0.017 & 0.134 & 0.277 & 0.440 & 0.220 & 0.011 & 0.121 & 0.032 & 0.905 & 0.089 & 0.097 & 0.109 \\ 

    \hline
\end{tabular}
\end{adjustbox}
\end{table}
As a third example, we apply Next-Door analysis  to a  gastric cancer  dataset, consisting of measurements on $p=2,200$ proteins,  from each of $12,480$ pixels (observations) obtained from 14 patients.
These data are  presented in  \cite{eberlin2014molecular}. In this example, instead selecting the model with the smallest randomized CV error, we use the CV one standard error rule. The CV folds are the same as the patients' id. The outcome is cancer ($Y=1$) versus normal ($Y=0$), and we fit a lasso-regularized logistic regression. The errors are based on the deviance from the fitted model. The results are shown in Table \ref{tab:gast}.  We select 19 proteins in the base model and 28 proteins in total are selected for all 19 proximal models.  Among the 19 proteins in the base model, we keep only those who has at least one p-value no greater than 0.05(15 left). Among the 28 proteins, we keep in the rows only those proteins whose coefficients' magnitude is at least 0.05 (20 left) , to save the space. The  model p-values suggest the first 6 proteins(\#487, \#476, \#607, \#431, \#1049, \#552 ) can be important to the models' predictive power with a p-value cut-off being 0.1.  The protein \#1509 is on the boundary(model p-value being 0.127), it might also be important as its proximal model has de-biased cv error larger than that of three other selected proteins.  In this example, because of the heterogeneity of $(\bx, y)$ from different patients(14 patients in the training  data and 5 patients in the test data),  the alignment between the test error and the CV error is not as good as the previous two examples.
\begin{table}[h]
\centering
\caption{\em Gastric Cancer  Data.  The leftmost column shows the fitted model lasso, and the remaining columns show the proximal model corresponding to the removal of each predictor.}
\label{tab:gast}.
\begin{adjustbox}{width = \textwidth, height = .28\textwidth}
\begin{tabular}{lrrrrrrrrrrrrrrrr}
  \hline
 &base & 487 & 476 & 607 & 1509 & 431 & 1049 & 552 & 1648 & 608 & 1374 & 606 & 1453 & 423 & 894 & 171 \\ 
  \hline
487 & 0.578 &  & 0.555 & 0.584 & 0.589 & 0.639 & 0.618 & 0.563 & 0.576 & 0.583 & 0.580 & 0.591 & 0.582 & 0.569 & 0.560 & 0.607 \\ 
  476 & 0.339 & 0.272 &  & 0.356 & 0.347 & 0.354 & 0.338 & 0.397 & 0.333 & 0.352 & 0.334 & 0.326 & 0.335 & 0.347 & 0.364 & 0.347 \\ 
  607 & 0.165 & 0.188 & 0.185 &  & 0.154 & 0.162 & 0.184 & 0.186 & 0.165 & 0.192 & 0.166 & 0.194 & 0.167 & 0.175 & 0.182 & 0.168 \\ 
  1509 & -0.206 & -0.213 & -0.219 & -0.196 &  & -0.223 & -0.207 & -0.201 & -0.207 & -0.204 & -0.206 & -0.205 & -0.232 & -0.212 & -0.193 & -0.205 \\ 
  431 & 0.244 & 0.527 & 0.258 & 0.242 & 0.278 &  & 0.220 & 0.245 & 0.238 & 0.241 & 0.229 & 0.265 & 0.249 & 0.238 & 0.360 & 0.220 \\ 
  1049 & 0.198 & 0.314 & 0.193 & 0.213 & 0.200 & 0.181 &  & 0.237 & 0.205 & 0.208 & 0.199 & 0.207 & 0.196 & 0.200 & 0.259 & 0.201 \\ 
  552 & 0.200 & 0.174 & 0.241 & 0.213 & 0.196 & 0.201 & 0.227 &  & 0.201 & 0.215 & 0.205 & 0.200 & 0.202 & 0.205 & 0.194 & 0.204 \\ 
  1648 & 0.064 & 0.041 & 0.047 & 0.066 & 0.067 & 0.055 & 0.081 & 0.068 &  & 0.061 & 0.072 & 0.076 & 0.064 & 0.064 & 0.007 & 0.058 \\ 
  1038 & 0.144 & 0.091 & 0.137 & 0.127 & 0.168 & 0.201 & 0.188 & 0.129 & 0.159 & 0.137 & 0.154 & 0.162 & 0.155 & 0.147 & 0.155 & 0.146 \\ 
  551 & 0.061 & 0.131 & 0.033 & 0.081 & 0.073 & 0.087 & 0.054 & 0.114 & 0.065 & 0.069 & 0.063 & 0.066 & 0.061 & 0.062 & 0.073 & 0.058 \\ 
  608 & 0.083 & 0.101 & 0.100 & 0.115 & 0.080 & 0.080 & 0.098 & 0.112 & 0.082 &  & 0.087 & 0.089 & 0.084 & 0.086 & 0.081 & 0.084 \\ 
  475 & 0.085 & 0.078 & 0.294 & 0.071 & 0.051 & 0.063 & 0.051 & 0.047 & 0.098 & 0.077 & 0.096 & 0.105 & 0.083 & 0.145 & 0.095 & 0.075 \\ 
  1596 & 0.021 &  & 0.074 & 0.009 & 0.011 & 0.108 & 0.001 & 0.022 & 0.057 & 0.015 & 0.033 & 0.014 & 0.013 & 0.019 & 0.082 & 0.033 \\ 
  1374 & 0.050 & 0.064 & 0.035 & 0.052 & 0.048 & 0.031 & 0.053 & 0.063 & 0.056 & 0.056 &  & 0.057 & 0.049 & 0.048 & 0.038 & 0.049 \\ 
  606 & 0.098 & 0.173 & 0.065 & 0.160 & 0.093 & 0.128 & 0.122 & 0.100 & 0.109 & 0.108 & 0.107 &  & 0.095 & 0.096 & 0.039 & 0.092 \\ 
  1453 & -0.043 & -0.084 & -0.028 & -0.050 & -0.171 & -0.055 & -0.036 & -0.053 & -0.043 & -0.046 & -0.043 & -0.039 &  & -0.050 & -0.055 & -0.043 \\ 
  423 & 0.088 & 0.019 & 0.113 & 0.113 & 0.112 & 0.078 & 0.090 & 0.110 & 0.088 & 0.093 & 0.085 & 0.086 & 0.092 &  & 0.061 & 0.092 \\ 
  894 & 0.242 & 0.166 & 0.267 & 0.257 & 0.226 & 0.298 & 0.294 & 0.240 & 0.228 & 0.243 & 0.240 & 0.226 & 0.245 & 0.234 &  & 0.244 \\ 
  171 & 0.035 & 0.209 & 0.051 & 0.040 & 0.031 & 0.002 & 0.043 & 0.046 & 0.030 & 0.036 & 0.034 & 0.029 & 0.034 & 0.039 & 0.040 &  \\ 
  898 &  &  &  &  &  &  &  &  &  &  &  &  &  &  & 0.059 &  \\ \hline
  cv\_error & 0.862 & 0.910 & 0.903 & 0.876 & 0.874 & 0.873 & 0.872 & 0.872 & 0.865 & 0.863 & 0.861 & 0.860 & 0.858 & 0.853 & 0.849 & 0.848 \\ 
  debiased\_error & 0.862 & 0.910 & 0.903 & 0.876 & 0.874 & 0.873 & 0.873 & 0.872 & 0.864 & 0.863 & 0.861 & 0.860 & 0.858 & 0.853 & 0.848 & 0.848 \\ 
    test\_error & 0.501 & 0.480 & 0.516 & 0.504 & 0.505 & 0.506 & 0.507 & 0.519 & 0.502 & 0.507 & 0.506 & 0.496 & 0.504 & 0.495 & 0.513 & 0.500 \\ 
   \hline
  selection frequency &  & 0.700 & 0.650 & 0.575 & 0.800 & 0.550 & 0.425 & 0.550 & 0.625 & 0.550 & 0.550 & 0.600 & 0.550 & 0.525 & 0.350 & 0.400 \\ 
  model pvalue &  & 0.033 & 0.001 & 0.013 & 0.127 & 0.046 & 0.091 & 0.000 & 0.814 & 0.812 & 0.745 & 0.648 & 0.532 & 0.893 & 0.992 & 0.959 \\ 
  model score &  & 0.047 & 0.001 & 0.023 & 0.159 & 0.085 & 0.214 & 0.000 & 1.302 & 1.476 & 1.355 & 1.079 & 0.968 & 1.700 & 2.833 & 2.397 \\ 
  feature pvalue &  & 0.000 & 0.000 & 0.003 & 0.000 & 0.000 & 0.000 & 0.000 & 0.000 & 0.000 & 0.000 & 0.000 & 0.005 & 0.000 & 0.000 & 0.000 \\  \hline
\end{tabular}\end{adjustbox}
\end{table}
\section{Extensions}
\label{sec:extension}
\subsection{Generalization to other supervised learning algorithms}
The methods proposed here can be extended in a straightforward manner to Cox's proportional hazards model and  the class of generalized linear models where the outcome $Y$ depends on a parameter vector $\eta$:
\begin{equation}
\eta =X\beta
\label{eq:model2}
\end{equation}
In this case, we have the penalized negative log likelihood as the objective function
\begin{equation}
J(\beta) = -\ell( \beta) + \lambda |\beta|.
\label{eq:obj2}
\end{equation}
The event $A_1$ can be characterized using the corresponding new CV loss we are interested in, and the selection frequency $\gamma_j$ remains unchanged.  As a result, the p-value neglecting the selection $A_2$ and the model score are both easily obtained in more complicated scenarios where we do not know how to characterize the selection for event $A_2$ in an efficient way, even for a black box model.  The Gastric cancer data set is an example where we apply the Next-door analysis to a classification problem.

In this paper, we considered some assumptions under which the asymptotic joint normality holds for the randomized CV curve.  In practice, we observe that such a jointly normality usually hold approximately, and the Bootstrap p-value itself is also usually robust. 

\subsection{A model with better out-of-sample performance}
The model p-value and model score can serve as an alternative feature importance measure even when we considers features only in the current selected feature set $E$. It provides a different ordering of feature importance compared with p-value, correlation or partial correlation with the response. It  can work better sometimes in practice as it considers the out of sample error directly.   We provide examples using the prostate data, HIV data and the gastric cancer data. We consider the selected feature set in the first stage and retrain the model using the training data with OLS/logistic regression. We build a sequence of nested model where we add feature one by one according to their model p-value, model score and feature p-value. For the prostate data set and gastric cancer data set, we start from models containing one feature. For the HIV data set, we start from models containing 2 features as the test errors are much larger for models with only one feature compared with the others. In Figure \ref{fig:performance}, we evaluate the models out of sample performance in the test set as a function of the number of features added.  The vertical dashed line is the number where we want to stop based on the model p-value.   In all three cases, the model p-values produced more sparse models with near optimal performance(smallest test errors achieved using the nested procedure).

\begin{figure}[!h]
\begin{center}
\includegraphics[width=.32\textwidth, height = .35\textwidth]{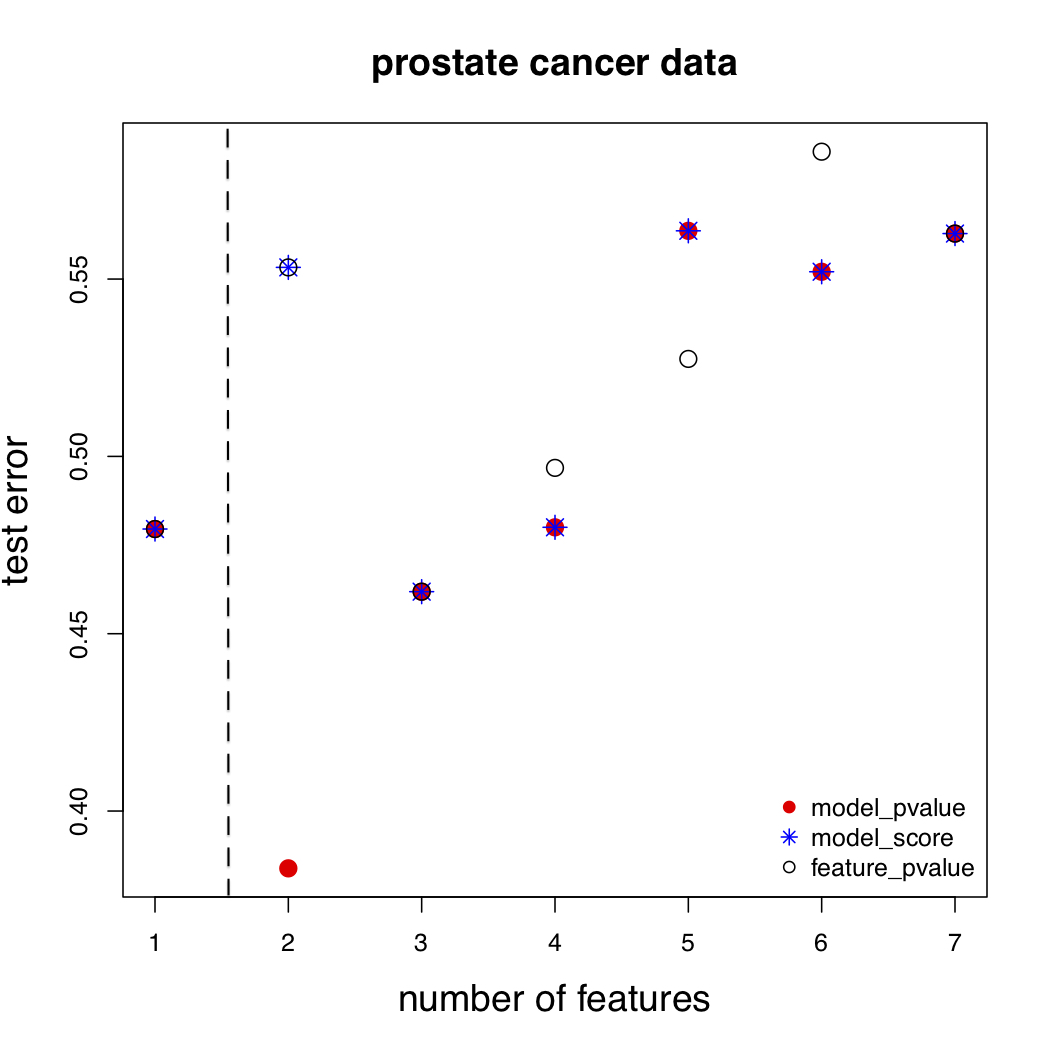}
\includegraphics[width=.32\textwidth, height = .35\textwidth]{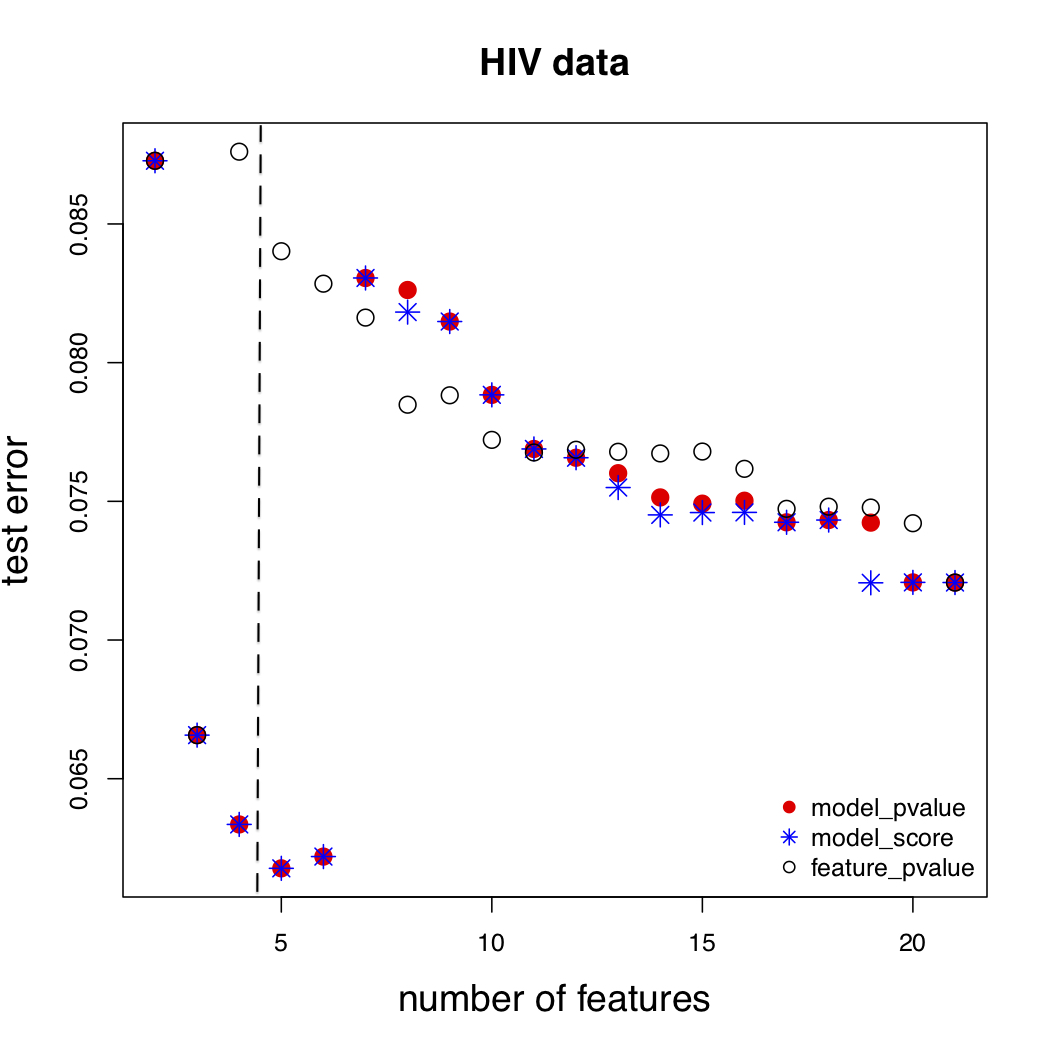}
\includegraphics[width=.32\textwidth, height = .35\textwidth]{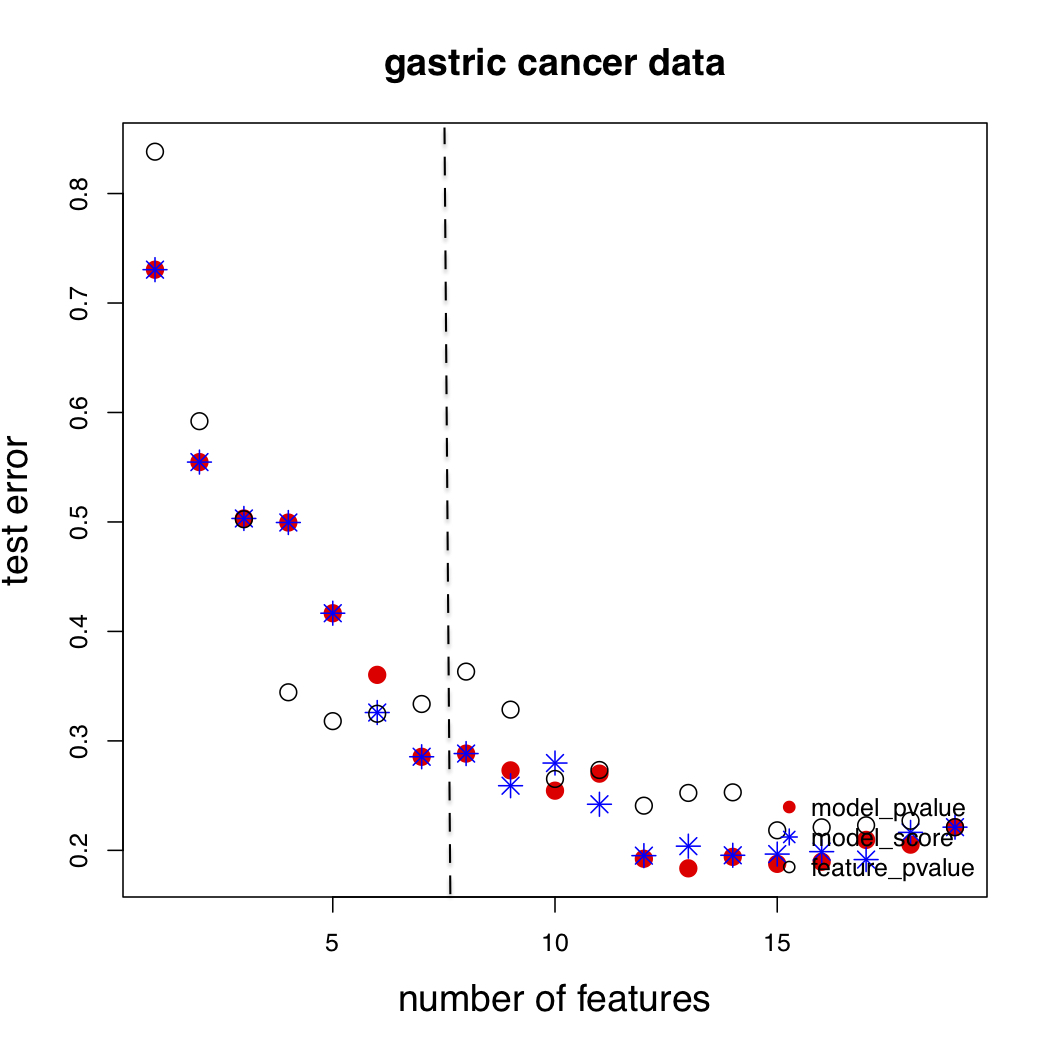}
\label{fig:performance}
\caption{Test errors for nested model sequences created based on model p-value, model score and feature p-value. The vertical dashed line is where we want to stop based on the mode p-values as described in section \ref{sec:realdata}.}
\end{center}
\end{figure}

\section{Discussion}
Our post-fitting procedure {\em Next-Door analysis}  gives insights into predictor indispensability and  offers nearby alternative models. Our proposal shifts the focus from coefficients to models:  having selected a model from the data,  we look for  alternative models that omit each predictor, and yet have validation error similar to the base model.  The model performance is considered marginalizing out the parameter turning and randomness in the training.  We present a  bootstrap approach based on the de-biased test error estimate for a pre-fixed hypothesis.  We also propose a simple concept called model score which takes into the hypothesis selection by considering its selection frequency.  By considering the hypothesis selection and model selection separately in this paper, we can easily deal with more complicated model selection and hypothesis selection events.

Next-Door analysis can also be used in cases where you want examine the removal of a set of predictors. In the case where the users have in mind which $k$ predictors they do not want to use after looking at the fitted model, we can simply remove those predictors and all analyses still carry through. In general, however, it is not practical to enumerate all different combinations of $k$ predictors. 

\vskip 14pt
\noindent {\large\bf Acknowledgements} The author would like to thank Professor Jonathan Taylor and  Professor Ryan Tibshirani for the helpful discussions.  The author would also like to thank Zhou Fan for his feedback on the paper. Robert Tibshirani was supported by NIH grant 5R01 EB001988-16 and NSF grant 19 DMS1208164.
\par

\appendix
\section{Simulation results for the Bootstrap p-values}
\label{app:bootstrap}
In this simulation, we examine the accuracy of the proposed bootstrap p-value and compare it to a naive bootstrap with the unadjusted errors. Each column of $X$ represents errors of a model with $n$ samples. We let $Q_j = \frac{\sum^n_{i=1} X_{i,j}}{n}$, construct the randomized error $Q^{\alpha}_j$,  and $Q^{\frac{1}{\alpha}}_j$ as described in section \ref{subsec:random}. Without loss of generality, we let $n = 100$.

Suppose that we have observations $X_{i,j} \sim \mathcal{N}(\mu_j, 1)$, $j = 1,2,\ldots, m$, $i = 1,2,\ldots, n$. Let $m = 5, 20$. For each $n$, we consider two cases for the underlying $\mu_j$: (1)$\mu_j = 0, \;\forall j = 1,2,\ldots m$, (2) $\mu_j  \sim \mathcal{N}(0, \frac{1}{n}) \;\forall j = 1,2,\ldots m$.   For each set of parameter, we repeat $10000$ times the following steps:
\begin{enumerate}
\item Construct  the de-biased estimate $\hQ^{\alpha}$ as described in  section \ref{subsec:random}, and the observed error estimate $\hQ  = \frac{1}{H}\sum^H_{h=1} Q_{k^*_h}$, where $k^*_h$ is the chosen index at $h^{th}$ round in the de-biased error estimate algorithm.
\item Bootstrap $B = 1000$ times, with or without mean rescaling. At each repetition $b$, let $\hQ_b$ and $\widehat{Q}^\alpha_b$ be the bootstrap version of the mean error and mean de-biased error, and the bootstrap differences are 
$$s_{1, b} = \hQ_b - \bar{Q},\;\; s_{2, b} = \hQ_b^\alpha- \bar{Q},$$
where $\bar{Q}$ are the Bootstrap population mean across repetitions.
\item Let  $\bar{\mu}$ be the true population mean marginalized over the given selection criterion. The p-values using the unadjusted error and the de-biased errors are given by:
$$p_{1, l} = \frac{\sum^B_{b=1} \mathbbm{1}_{\{\hQ -  \bar{\mu} \geq s_{1, b}\}}}{B}, \;\;p_{1, r} = \frac{\sum^B_{b=1} \mathbbm{1}_{\{\hQ -\bar{\mu} \leq s_{1, b}  \}}}{B}$$
$$p_{2, l} = \frac{\sum^B_{b=1} \mathbbm{1}_{\{\widehat{Q}^{\alpha}-  \bar{\mu}\geq s_{2, b}\}}}{B}, \;\;p_{2, r} = \frac{\sum^B_{b=1} \mathbbm{1}_{\{\widehat{Q}^{\alpha}-\bar{\mu}\leq s_{2, b}  \}}}{B}$$
Here, $p_{1, l}$ and $p_{2, l}$ are the probability of the bootstrap differences between the estimate and truth being smaller than the difference between our current estimate and the underlying truth, if $p_l < \alpha$, it means that the truth will be on the left to the confidence interval constructed for a given level $\alpha$; similarly, if $p_r< \alpha$, it means that the truth will be on the right of the confidence interval constructed for a given level $\alpha$.
\end{enumerate}
In this simulation, we know that the test statistics is not degenerate, so we let $\gamma_2 = 0$ in the Bootstrap algorithm and we know that the covariance is not degenerate, so we let $\gamma_1 = 0$ in the de-biased error estimate algorithm. In Figure \ref{fig:coverage}, the left and right halves  show  the empirical CDF plot of the four p-values after mean rescaling and without mean rescaling across four parameter settings.  
\begin{figure}
\centering
\includegraphics[height = .45\textwidth]{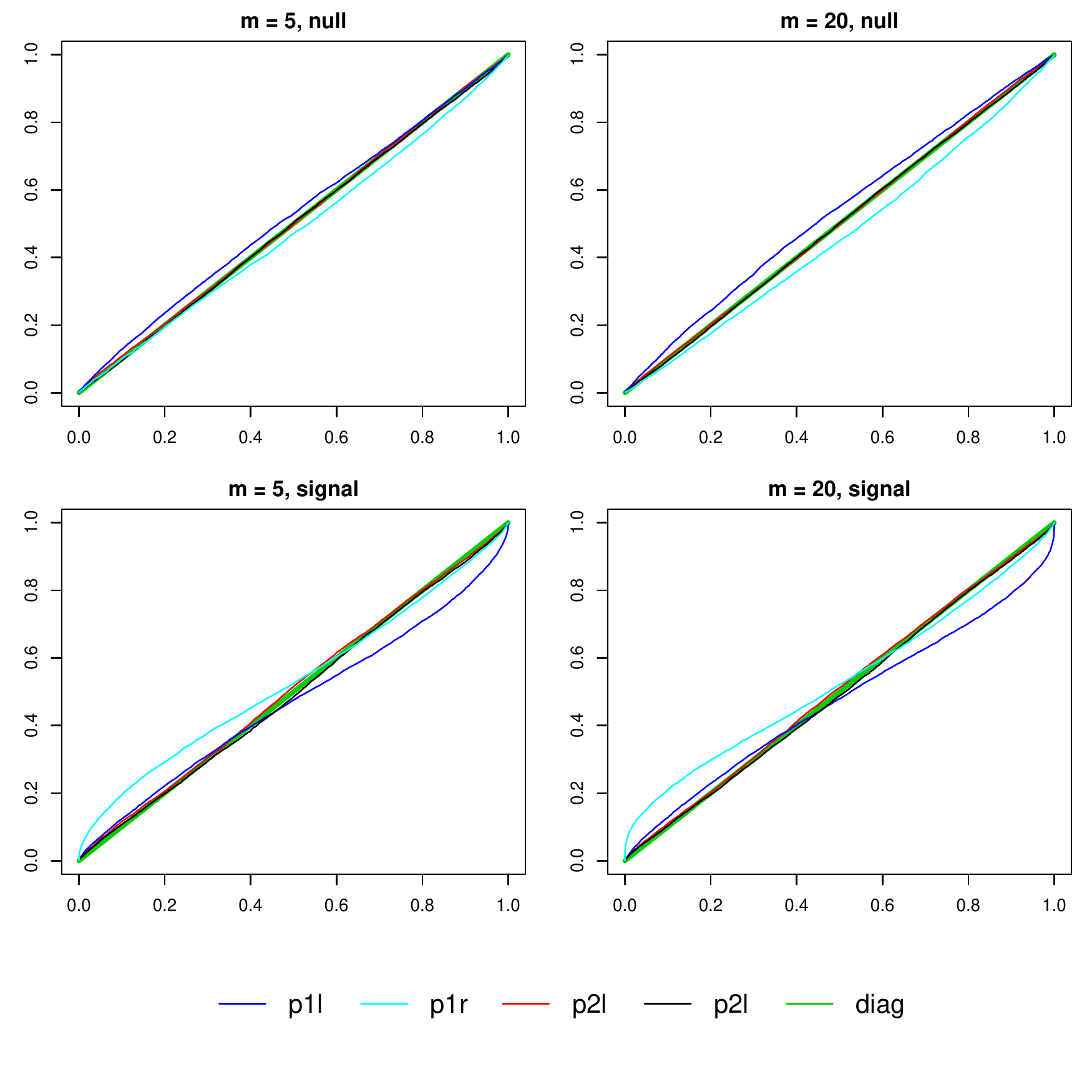}
\includegraphics[height = .45\textwidth]{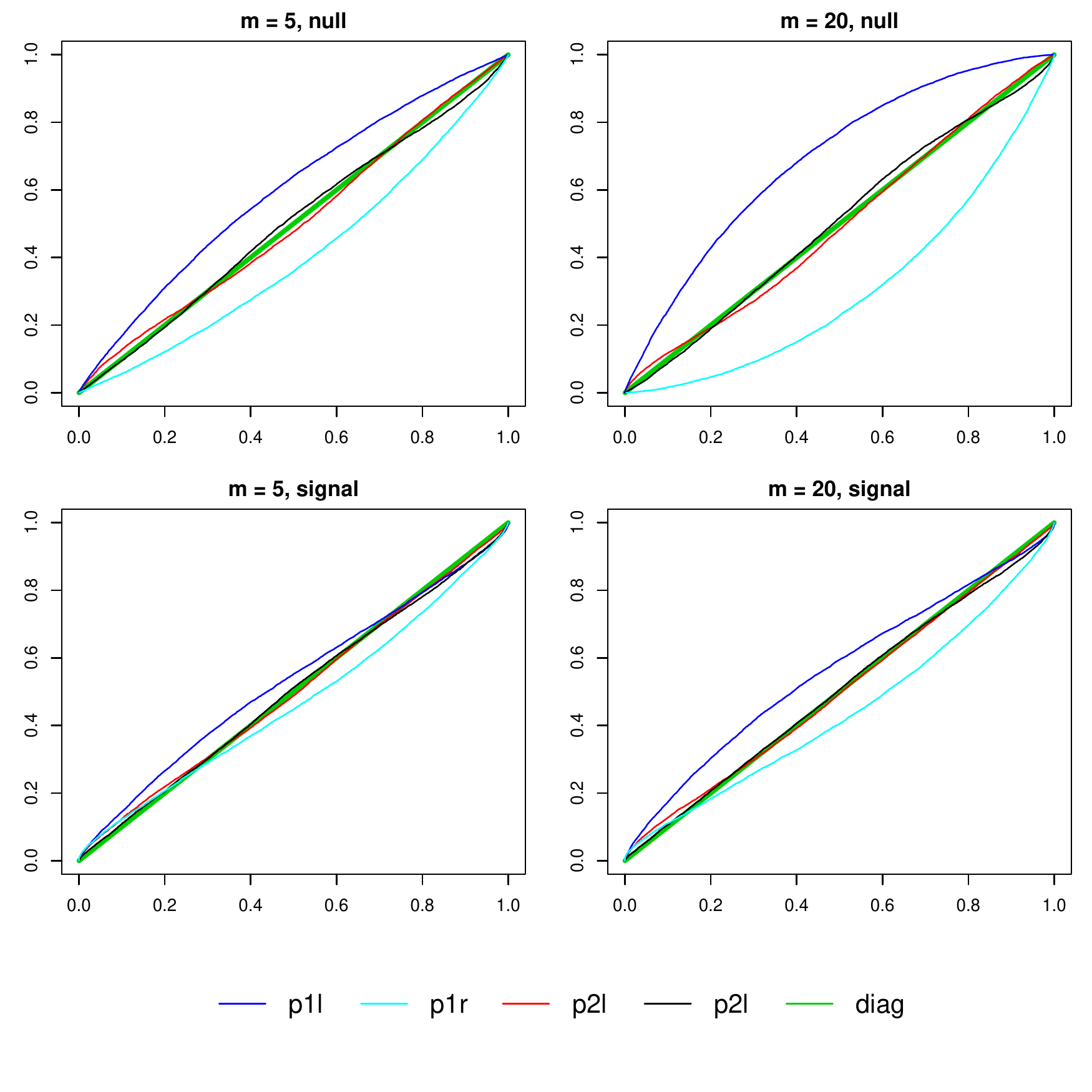}
\vskip -.2in
\caption{\em p-value distribution after 10000 repetitions. The left plot is the distribution of four p-values after variance rescaling and the right one shows the p-values distribution without variance rescaling. The light/dark blue curves show the c.d.f of the p-values testing for the observed value comes from a distribution with mean greater/smaller than the true mean using the unadjusted observation; the red/black curves show that with the de-biased estimate.}
\label{fig:coverage}
\end{figure}
The de-biased estimate Bootstrap has p-value distribution closer to uniform -- it has less dependence on the correct underlying distribution than the native Bootstrap procedure. Also, the mean rescaling approach leads to better p-value distribution.

\section{Proofs of  Theorem \ref{thm:randomized0}, \ref{thm:conservative} , Lemma \ref{lem:covariance}} 
\label{app:proof}
We will use  Proposition \ref{prop:normality}   to prove Theorem \ref{thm:randomized0} and  Lemma \ref{lem:covariance}. 
\begin{proposition}
\label{prop:normality} 
 Suppose the consistency, range, moment and dimension assumptions hold. For all $\lambda \in \Lambda$, let  $W(\bx_i, y_i) = (y_i- \bx^T_i\beta)^2$  and $\delta_{i} = (y_i- \bx^T_i\hat{\beta})^2 - W(\bx_i, y_i)$ where $\hbeta$ is the coefficients from  the model fitted with a training set of  size $n$ and $\beta$ is the coefficient vector it converges in the consistency assumption. Let $\{(\bx_i, y_i), \;i = 1,2,\ldots, n_2 \}$ be a test set of size $n_2\asymp n$, then there exists a constant large enough, such that,
\[
\frac{1}{\sqrt{n_2}}\sum^{n_2}_{i=1} (W(\bx_i, y_i) - \mu)\overset{d}{\rightarrow} N(0, \sigma^2), \;\;E[(\frac{1}{\sqrt{n_2}}\sum^{n_2}_{i=1} \delta_{i})^2]\rightarrow 0,\;\;E[\frac{1}{n_2}\sum^{n_2}_{i=1} \delta^2_{i}]\rightarrow 0
\]
with $\mu = E[(y - \bx^T\beta)^2]$ and $\sigma^2 = \Var((y - \bx^T\beta)^2)\leq C$  as $n, \;n_2\rightarrow \infty$.
\end{proposition}
\begin{proof}
The first statement is a direct application of CLT. For the second and the third statements, we divide $\frac{1}{\sqrt{n_2}}\sum^{n_2}_{i=1}\delta_{i}$ into and bound $\frac{1}{n_2}\sum^{n_2}_{i=1}\delta^2_{i}$ by two parts:
\begin{align*}
&\frac{1}{\sqrt{n_2}}\sum^{n_2}_{i=1}\delta_{i} = \frac{1}{\sqrt{n_2}}\sum^{n_2}_{i=1}2(y_i-\bx^T_i\beta)\bx^T_i(\beta-\hat{\beta})+ \frac{1}{\sqrt{n_2}}\sum^{n_2}_{i=1}[\bx^T_i(\beta-\hat{\beta})]^2\\
&\frac{1}{n_2}\sum^{n_2}_{i=1}\delta^2_{i} \leq \frac{1}{n_2}\sum^{n_2}_{i=1}8[(y_i-\bx^T_i\beta)\bx^T_i(\hat{\beta}- \beta)]^2+ \frac{1}{n_2}\sum^{n_2}_{i=1}2[\bx^T_i(\beta-\hat{\beta})]^4
\end{align*}
%When the consistency consumption and moment assumptions hold, 
\begin{itemize}
\item By Cauchy-Schwarz inequality, we have $E[(\frac{1}{\sqrt{n_2}}\sum^{n_2}_{i=1}(y_i-\bx^T_i\beta)\bx^T_i(\hat{\beta}- \beta))^2]\leq E[n^{-1/2}_2\|\sum^{n_2}_{i=1}(y_i-\bx^T_i\beta)\bx^T_i\|^2_2]E[\sqrt{n_2}\|\beta-\hat{\beta}\|^2_2]$($\hat{\beta}$ and $(y_i, \bx_i)$ are independent).  The second half $E[\sqrt{n_2}\|\beta-\hat{\beta}\|^2_2]\rightarrow 0$ by the consistency assumption. For the first term, we have
\begin{align*}
E[n^{-1/2}_2(\sum^{n_2}_{i=1}(y_i-\bx^T_i\beta)x_{i,j})^2] &= E[n^{-1/2}_2(\sum^{n_2}_{i=1}((y_i-\bx^T_i\beta)x_{i,j} - E[(y-\bx^T\beta)x_j])+n_2E[(y-\bx^T\beta)x_j])^2]\\
& = \Var((y_i-\bx^T_i\beta)x_{i,j}) + n^{1/2}_2(E[(y-\bx^T\beta)x_j])^2
\end{align*}
$\Var((y_i-\bx^T_i\beta)x_{i,j})$ is bounded because $E[(y-\bx^T\beta)^2\|\bx\|_2^2]$ is bounded. Now, we show that, when the range assumption holds, $n^{\frac{1}{4}}|E[(y-\bx^T\beta)x_j]| \leq C$ for a large enough constant $C$. Imagine now we have another training data set $(\tilde{\bold{x}}_i, \ty_i)$, with  sample size $N$, and let $\tbeta$ be the coefficient trained using this dataset. By the KKT condition, we have $|\frac{1}{N}\sum_{i=1}\tx_{i,j}(\ty_i - \tbx_i^T\tbeta)| \leq \lambda$, hence, the following is true
\begin{align*}
|\frac{1}{N}\sum_{i=1}\tx_{i,j}(\ty_i - \tbx_i^T\beta)| \leq |\frac{1}{N}\sum_{i=1}\tx_{i,j} \tbx_i^T(\tbeta -\beta)|+\lambda \leq \lambda+  \|\frac{1}{N}\sum_{i=1}\tx_{i,j}\tbx_i^T\|_2\|\tbeta -\beta\|_2
\end{align*}
We know that $N$ can be arbitrarily large and grow to $\infty$ faster than $n$. By LLN, for any fixed $n$, we let $N\rightarrow\infty$ $\lim_{N\rightarrow \infty}(\frac{n^{\frac{1}{4}}}{N}\sum_{i=1}\tx_{i,j}(\ty_i - \tbx_i^T\beta) - n^{\frac{1}{4}}E[(y-\bx^T\beta)x_j]) =  0$ and $ \|\frac{1}{N}\sum^N_{i=1}\tx_{i,j}\tbx_i^T\|_2$ converges to a finite constant in probability.  By the consistency assumption,  we have $N^{\frac{1}{4}}\|\tbeta -\beta\|_2\overset{p}{\rightarrow} 0$. By the range assumption, for a large enough constant $C$, we have 
\[
\lim_{n\rightarrow\infty} n^{\frac{1}{4}}|E[(y-\bx^T\beta)x_j]|  \leq \lim_{n\rightarrow\infty} \lim_{N\rightarrow\infty}(\lambda n^{\frac{1}{4}}+  \|\frac{1}{N}\sum^N_{i=1}\tx_{i,j}\tbx_i^T\|_2\|\tbeta -\beta\|_2n^{\frac{1}{4}})\leq C
\]
 Hence, we have $E[(\frac{1}{\sqrt{n_2}}\sum^{n_2}_{i=1}(y_i-\bx^T_i\beta)\bx^T_i(\hat{\beta}- \beta))^2]\rightarrow 0$.  
\item By LLN, the moment assumption and consistency assumption, we have $\frac{1}{n_2}E[\sum^{n_2}_{i=1}((y_i-\bx^T_i\beta)\bx^T_i(\hat{\beta} - \beta))^2] \leq  E[\frac{1}{n_2}\sum^{n_2}_{i=1}\|(y_i-\bx^T_i\beta)\bx_i\|_2^2]E[\|\hat{\beta} - \beta\|_2^2]=E[(y-\bx^T\beta)^2\|\bx\|_2^2]E[\|\hat{\beta} - \beta\|_2^2]\rightarrow 0$. 
\item By Cauchy-Schwarz inequality and LLN again, we see that the second terms in both expressions go to 0:
\begin{align*}
&E[\frac{1}{n_2}\sum^{n_2}_{i=1}(\bx^T_i(\beta-\hat{\beta}))^4] \leq  E[\|\bx\|_2^4]E[\|\beta-\hat{\beta}\|^4_2] \rightarrow 0\\
&E[(\frac{1}{\sqrt{n_2}}\sum^{n_2}_{i=1}[\bx^T_i(\beta-\hat{\beta})]^2)^2]\leq E[\sum^{n_2}_{i=1}(\bx^T_i(\beta-\hat{\beta}))^4=E[\|\bx\|_2^4](n_2E[\|\beta-\hat{\beta}\|^4_2])\rightarrow 0, 
\end{align*}
\end{itemize}
Hence, we have  $E[(\frac{1}{\sqrt{n_2}}\sum^{n_2}_{i=1} \delta_{i})^2]\rightarrow 0,\;\;E[\frac{1}{n_2}\sum^{n_2}_{i=1} \delta^2_{i}]\rightarrow 0$.  
\end{proof}
The above results also hold if we exclude predictor $j$. Let $W_k(\bx_i, y_i)= (y_i- \bx^T_i\beta(\lambda_k))^2$ and $W_{m+k}(\bx_i, y_i) =  (y_i- \bx^T_i\beta(\lambda_k; j))^2$ for $k = 1,2,\ldots, m$. Let $\delta_{k,i} = Q_{k}(\bx_i, y_i) -W_{k}(\bx_i, y_i)$ and $W_k, \delta_k$ be their mean over samples $i$ for $k=1,\ldots, 2m$. Let $\mu_k$ be the mean  of $W_k(\bx_i, y_i)$. As a direct result of Proposition \ref{prop:normality},  let $Q^v_k$ be the mean validation error at penalty $k$ and fold $v$, for a constant $C$ large enough, we have 
\begin{align}
\label{eq:eq1}
&\Var(Q_k(\bx_i, y_i)) \leq 2\Var(W_k(\bx_i, y_i))+E[\delta_{k,i}^2] \leq C,\;\;|E[Q_k] - E[W_k(\bx, y)]| \rightarrow 0\nonumber\\
&\Var(Q^v_k) \leq 2\Var(W_k(\bx, y))+2E[(\frac{V}{\sqrt{n}}\sum^{n}_{i=1}\delta_{k,i})^2]  \leq C
\end{align}
\noindent\textbf{Proof of Lemma \ref{lem:covariance}: } Let $\Sigma_{k,k'} =\Cov(W_k(\bx, y), W_{k'}(\bx, y))$. We divide the covariance estimate into three parts:
\begin{align*}
\hat{\Sigma}_{k, k'} &= \frac{\sum^n_{i=1}(Q_k(\bx_i , y_i) - Q_k)(Q_{k'}(\bx_i , y_i) - Q_{k'}) }{n}\\
& = \frac{\sum^n_{i=1}(W_k(\bx_i , y_i) - W_k)(W_{k'}(\bx_i , y_i) - W_{k'})}{n}+\frac{\sum^n_{i=1}(\delta_{i,k}-\delta_k)(\delta_{i,k'}-\delta_{k'})}{n} \\
&+ \frac{\sum^n_{i=1}(\delta_{i,k}-\delta_k)(W_k'(\bx_i , y_i) - W_k')+\sum^n_{i=1}(\delta_{i,k'}-\delta_{k'})(W_k(\bx_i , y_i) - W_k)}{n}
\end{align*}
\begin{itemize}
\item  The term  $|\frac{\sum^n_{i=1}(W_k(\bx_i , y_i) - W_k)(W_{k'}(\bx_i , y_i) - W_{k'})}{n} - \Sigma_{k,k'}|\overset{p}{\rightarrow} 0$ by LLN.
\item By Slutsky's theorem, Proposition \ref{prop:normality} and use the fact that $W_{k}(\bx_i, y_i)$ has finite variance, we have
\begin{align*}
&|\frac{\sum^n_{i=1}(\delta_{i,k}-\delta_k)(W_{k'}(\bx_i , y_i) - W_{k'})+\sum^n_{i=1}(\delta_{i,k'}-\delta_{k'})(W_k(\bx_i , y_i) - W_k)}{n}|\\
\leq & \sqrt{\frac{\sum^n_{i=1}(W_{k'}(\bx_i , y_i) - W_{k'})^2}{n} \frac{\sum^n_{i=1}(\delta_{k,i} - \delta_k)^2}{n}} + \sqrt{\frac{\sum^n_{i=1}(W_{k}(\bx_i , y_i) - W_{k})^2}{n} \frac{\sum^n_{i=1}(\delta_{k,i} - \delta_k)^2}{n}}\overset{p}{\rightarrow} 0
\end{align*}
and $|\frac{\sum^n_{i=1}(\delta_{i,k}-\delta_k)(\delta_{i,k'}-\delta_{k'})}{n}| \leq  \sqrt{\frac{\sum^n_{i=1} (\delta_{k,i} - \delta_k)^2}{n} \frac{\sum^n_{i=1}(\delta_{k',i} - \delta_{k'})^2}{n}}\overset{p}{\rightarrow}  0$.
\end{itemize}
As a result, we have $\|\hat{\Sigma} - \Sigma\|_\infty \overset{p}{\rightarrow} 0$. Now we show that that  $E[\sum_{k} \hat{\Sigma}_{k,k}] < C$ for a large constant $C$.  Let  $\hat{\beta}^v(\lambda_k)$ be the coefficient trained for predicting fold $k$, we have
\begin{align*}
E[\hat{\Sigma}_{k,k}]= \frac{1}{V}\sum^V_{v=1} E[\frac{\sum_{i\in \mathcal{V}_v} (Q_k(\bx_i, y_i) - Q^v_k+Q^v_k-Q_k)^2}{n/V}] \leq \frac{2}{V}\left(\sum^V_{v=1}\Var((y-\bx\hbeta^v(\lambda_k))^2)+\sum^V_{v=1}\Var(Q^v_k)\right)
\end{align*}
Bothe the first term  and the second terms are bounded by equation (\ref{eq:eq1}). We thus prove that  $E[\sum_{k}\hat{\Sigma}_{k,k}]$ is bounded.

\textbf{Proof of Lemma \ref{lem:expectation}:} Let $\hbeta^v(\lambda_{k^*})$ be the coefficients trained for fold $v$ at penalty $\lambda_{k^*}$. By definition $\sqrt{n}(\Err^R - \sum^m_{k^*=1} \Err_{k^*} P(O_{k^*})) =\sqrt{n}E[\sum^m_{k^*=1}(\frac{1}{V}\sum^V_{v=1}(y - \bx^T\hbeta^v(\lambda_{k^*}))^2-\Err_{k^*}) \mathbbm{1}_{O_{k^*}}]$. Because both $V$ and $m$ are finite, we only need to show that for each $k^*$ and $v$, we have $|\sqrt{n}E[((y - \bx^T\hbeta^v(\lambda_{k^*}))^2-\Err_{k^*}) \mathbbm{1}_{O_{k^*}}]|\rightarrow 0$.  Let $(\tbx_i, \ty_i)$ for $i = 1,2,\ldots, n$ be $n$ new realizations, we know that
\begin{align*}
|\sqrt{n}E[((y - \bx^T\hbeta^v(\lambda_{k^*}))^2-\Err_{k^*}) \mathbbm{1}_{O_{k^*}}]| & =|E[\frac{1}{\sqrt{n}}\sum^n_{i=1}((\ty_i-\tbx_i^T\hbeta^v(\lambda_{k^*}))^2 - W_{k^*}(\tbx_i - \ty_i))\mathbbm{1}_{O_{k^*}}]|
\end{align*}
Let $I_{k^*} := \frac{1}{\sqrt{n}}\sum^n_{i=1}((\ty_i-\tbx_i^T\hbeta^v(\lambda_{k^*}))^2 - W_{k^*}(\tbx_i - \ty_i))$. By Proposition \ref{prop:normality} ,  we know that $E[I^2_{k^*}] \rightarrow 0$. By the Cauchy-Schwarz inequality, $|\sqrt{n}E[((y - \bx^T\hbeta^v(\lambda_{k^*}))^2-\Err_{k^*}) \mathbbm{1}_{O_{k^*}}]|  \leq E[I^2]P(O_{k^*})\rightarrow 0$, and equivalently, $|\sqrt{n}(\Err^R - \sum^m_{k^*=1} \Err_{k^*} P(O_{k^*})) | \rightarrow 0$. Following exactly the same argument, we have $|\sqrt{n}(\Err^{j,R}- \sum^m_{k^*=1} \Err_{m+k^*} P(O_{k^*})) |\rightarrow 0$.
\subsection*{Proof of Theorem \ref{thm:randomized0}}
Besides Proposition \ref{prop:normality}, we will also use Proposition \ref{prop:prop2} below.
\begin{proposition}
\label{prop:prop2}(\cite{guan2018test})
Let m be a fixed number and $\{x_n\}$, $z$ be m dimensional vectors such that $z \sim N(0, \Sigma)$ and $x_n\overset{D}{\rightarrow} z$, and $E[\|x_n\|^2_2]$ is asymptotically bounded.  For a sequence of bounded function $g_n(.)$ that is almost everywhere differentiable with bounded first derivative under both the measures of $z$ and $x_n$ asymptotically, we have
\[
\lim_{n\rightarrow \infty}\|E[x_ng_n(x_n)] - E[zg_n(z)]\|_\infty = 0
\]
\end{proposition}
We apply Proposition \ref{prop:normality} to the cross-validation error curve, we have $\sqrt{n}(Q - \mu) \overset{d}{\rightarrow} N(0, \Sigma)$, where $\mu = (\mu_1,\mu_2,\ldots, \mu_{2m})$ and $\Sigma$ is bounded but potentially not invertible. By Proposition \ref{prop:normality}, we know $\sqrt{n}\|\Err - \mu\|_\infty \rightarrow 0$. As a consequence,  $\widetilde{Q}^{\alpha}$ and $\widetilde{Q}^{\frac{1}{\alpha}}$ are asymptotically independent with invertible covariance matrix when $\|\widehat{\Sigma} - \Sigma\|_{\infty} \overset{p}{\rightarrow} 0$:
\[
\sqrt{n}(\left(\begin{array}{l}\widetilde{Q}^{\alpha} \\\widetilde{Q}^{\frac{1}{\alpha}} \end{array}\right) - \left(\begin{array}{l}{\rm Err}\\{\rm Err} \end{array}\right)) \overset{d}{\rightarrow} \mathcal{N}(0, \left(\begin{array}{l l} (1+\alpha)(\Sigma+\sigma^2_0 I)& 0\\ 0&(1+\frac{1}{\alpha})(\Sigma+\sigma^2_0 I) \end{array}\right))
\]
Let $Z^{\alpha}$ and $Z^{\frac{1}{\alpha}}$ be the normal vectors from the limiting distribution corresponding to $\widetilde{Q}^{\alpha}$ and  $\widetilde{Q}^{\frac{1}{\alpha}}$.  The asymptotic independence guarantees that any selection using $\widetilde{Q}^{\alpha}$ has diminishing effect in $\widetilde{Q}^{\frac{1}{\alpha}}$:
\begin{align*}
\sqrt{n}E[\widehat{\Err} -\sum^m_{k^*=1} \Err_{k^*}\mathbbm{1}_{O_{k^*}}] &=E[\sum^m_{k^*=1}\sqrt{n}(\widetilde{Q}^{\frac{1}{\alpha}}_{k^*}(\epsilon, z)-\Err_{k^*})\mathbbm{1}_{O_{k^*}}] 
\end{align*}
In our case the vector $\tQ^{\alpha}-\Err$ and $\tQ^{\frac{1}{\alpha}}-\Err$ are both square integrable by equation (\ref{eq:eq1}) in Proposition \ref{prop:normality} and Lemma \ref{lem:covariance}:
\begin{align*}
&E[\|\sqrt{n}(\left(\begin{array}{l}\widetilde{Q}^{\alpha} \\\widetilde{Q}^{\frac{1}{\alpha}} \end{array}\right) - \left(\begin{array}{l}{\rm Err}\\{\rm Err} \end{array}\right))\|_2^2] \\
= &\Var(\sqrt{n}Q)+E[E[\|\epsilon\|_2^2|\hat{\Sigma}]]+(\alpha+\frac{1}{\alpha})E[E[\|z\|_2^2|\hat{\Sigma}]] \\
= &\Var(\sqrt{n}Q)+(1+\alpha+\frac{1}{\alpha})m\gamma_1\sigma^2_0+(\alpha+\frac{1}{\alpha})E[\sum^{2m}_{k=1}\hat{\Sigma}_{k,k}] < \infty
\end{align*}
Because $Z^{\alpha}$ has invertible covariance matrix and for each $\mathbbm{1}_{O_{k^*}}$, it is almost everywhere differentiable with first derivative being 0 under both the measures of $Z^{\alpha}$ and $\tQ^{\alpha}$.  Based on Proposition \ref{prop:prop2}, and the independence between  $Z^{\alpha}$ and $Z^{\frac{1}{\alpha}}$ , we have 
\[
\lim_{n\rightarrow\infty}E[\sum^m_{k^*=1}\sqrt{n}(\widetilde{Q}^{\frac{1}{\alpha}}_{k^*}(\epsilon, z)-\Err_{k^*})\mathbbm{1}_{O_{k^*}}] =\lim_{n\rightarrow\infty} E[\sum^m_{k^*=1}Z^{\frac{1}{\alpha}}_{k^*}(\epsilon, z)\mathbbm{1}_{\{Z^{\alpha}_{k^*}+\sqrt{n}\Err_{k^*} < Z^{\alpha}_{k}+\sqrt{n}\Err_{k} , \;\forall k\neq k^*\} } ] = 0
\]
Finally, we apply Lemma \ref{lem:expectation}, we have $\sqrt{n}(E[\widehat{\Err}]-\Err^{R}) \rightarrow 0$. Similarly, we have $\sqrt{n}(E[\widehat{\Err}^{j}]-\Err^{j, R}) \rightarrow 0$.
\section{A post selection inference approach conditional on the selected penalty}
\label{sec:postSelection}
While the Bootstrap p-value described in Section \ref{sec:test} tests for a quantity marginalizing out all randomness, the post selection inference approach conditional on the selected penalty $\lambda_{k^*}$. It corresponds to the case where in practice, we will fix the penalty selected from now on. The testing problem is then
\[
H_0: \Err^{j}_{k^*} \leq \Err_{k^*}   \;\;\;\;vs. \;\;\;\;H_1:\Err^{j}_{k^*}> \Err_{k^*} \nonumber ;\;\;\;\;\;\;\;\;\;\;\;\;\;\;\;\;\;\;\;\;\;\;\;\;\;\;\;\;\;\;\;\; (G_2)
\]

 In \cite{markovic2017adaptive}, let $\lambda_{k^*}$ be the selected penalty, the author have conditioned on  (1)the feature set $E$ is selected with penalty $\lambda_{k^*}$  and all predictors. (2)$\lambda_{k^*}$ is the penalty which minimizes the randomized CV curve, defined as
\[
\widetilde{Q}_k = Q_k+\frac{\epsilon_k}{\sqrt{n}}, \;\;\forall k = 1,2,\ldots, m
\]
where $\epsilon_k\sim \mathcal{N}(0, \tau^2)$ with $\tau^2$ being a constant. For all $\forall j\in \E$, the author then compare the performance between models fitted using OLS with feature set $E$ and $E\setminus j$, trained with all data.  In this section, we modify their procedure in the following three aspects: (1)For the model excluding the predictor $j$, we train it with all predictors except for $j$ instead of restricting ourselves to $E\setminus j$. (2)We do not run OLS with all data because we only care about the out of sample performance for  models produced. (3)We  neglect the selection event $A_2$ that $j\in S_{k^*}$ because we want to show that the conditional on the first event $A_1$ only will lead to  loss of power.

We look at  the test statistics $T =Q^j_{k^*}-Q_{k^*}+\frac{\epsilon}{\sqrt{n}}$, where $\epsilon \sim \mathcal{N}(0, \tau^2)$.  Let $\widetilde{Q} = (\widetilde{Q}_1, \widetilde{Q}_2, \ldots, \widetilde{Q}_m)^T$ . The event $A_1$ can be characterized by $H_{A_1} =\{\widetilde{Q}\in R^{m}: B_{Q}\widetilde{Q} \leq 0\}$, where $B_{Q}$ is the $m\times m$ matrix with a  1 and $-1$ at entry $(k, k)$ and $(k^*, k)$ for $k \neq k^*$, and $0$ at other entries. Let the $\Sigma$ be the covariance structure of the vector $\sqrt{n}(T, D, \tQ^T)^T$. We use $\Sigma_{TT}$ for $n\Cov(T, T)$ and $\Sigma_{T\tQ}$ for $n\Cov(T, \tQ)$, etc.

 Let  $\widetilde{Q} = \alpha_{\widetilde{Q} }T + N_{ \widetilde{Q}}$  where $\alpha_{\widetilde{Q} }:=\Sigma_{ \widetilde{Q}T} \Sigma^{-1}_{TT}$. The intuition is that if the three variables $\sqrt{n}T$,  and $\sqrt{n}\tQ$ are jointly asymptotically normal, then $N_{\tQ}$ is asymptotically independent of $T$. We can then condition on $N_{\tQ}$ and write the constraints in terms of $T$'s asymptotic behavior and achieve an asymptotic guarantee for the type $I$ error control.
 \begin{proposition}(\cite{markovic2017adaptive}, Theorem 1)\label{prop:post}Let $T$ be the test statistics.  If the following two assumptions hold

(1)The selection event $A$ can be characterized in terms of affine constraints over some data vector $D\in S_D =  \{D'|BD' \leq b\}$.

(2) The asymptotic joint normality of $(T, D)$  with invertible covariance matrix holds pre-selection
\[
\left(\left(\begin{array}{l} T\\ D\end{array}\right)\right)  \overset{d}{\rightarrow}  \mathcal{N}\left(  \left(\begin{array}{l} \theta \\ \bold{\gamma} \end{array}\right), \left(\begin{array}{l l l} \bSigma_{TT}& \bSigma_{TD}  \\ \bSigma_{DT}& \bSigma_{DD} \end{array}\right)\right)
\] 
Let $D = \Sigma_{D, T}\Sigma_{T,T}^{-1}T+N_D$,  $(Z_T, Z_{D})$ be the normal vectors from the limiting distribution,  then we have
\[
P_{\theta, D\in S_D}(\|Z_T-\theta\|_2\leq \|T-\theta\|_2|Z_D\in S_{D}, Z_{D}-\Sigma_{D, T}\Sigma_{T,T}^{-1}Z_T= N_{D})  \overset{d}{\rightarrow} \rm{Unif}[0,1]
\]
\end{proposition}
The proposition also works for the one side test. By Proposition \ref{prop:normality}, we know $T$ and $\tQ$ are asymptotically jointly normal with invertible covariance matrix. Hence, we can construct the p-value for any hypothesis value $\theta$ we are interested in based on Lemma \ref{lem:selectionRegion}:
 
 \begin{lemma}
\label{lem:selectionRegion}
Let $\theta$ be the  hypothesized mean of $T$ and $\sqrt{n}(Z_T-\theta)$ be the normal variable from the limiting distribution of $\sqrt{n}(T-\theta)$. Under the consistency, moment and dimension assumptions, we have
\begin{enumerate}
\item  The following construction of p-value achieving the asymptotic uniformity under the simple null hypothesis parameter $\theta$,
\begin{align*}
&p_{\theta} := P_{\theta,\widetilde{Q}\in  S_{A_1}, D\in S_{A_2} , N_D, N_{\tQ}}(Z_T \geq T| Z_T\in (a,  b))\overset{d}{\rightarrow}\rm{Unif}[0,1]
\end{align*}
with  $a= \underset{k:\alpha_{\widetilde{Q}, k^*}- \alpha_{\widetilde{Q},k} < 0}{\max} \frac{N_{\widetilde{Q}, k}- N_{\widetilde{Q}, k^*}}{\alpha_{\widetilde{Q}, k^*}- \alpha_{\widetilde{Q}, k}}$, $b = \underset{k:\alpha_{\widetilde{Q}, k^*}- \alpha_{\widetilde{Q},k} > 0}{\min} \frac{N_{\widetilde{Q}, k}- N_{\widetilde{Q}, k^*}}{\alpha_{\widetilde{Q}, k^*}- \alpha_{\widetilde{Q}, k}}$.
\item The type I error for the null can be controlled by controlling type I error at $\theta = 0$: $\lim_{n\rightarrow \infty}P_{H_0}(p_0 \leq \alpha) \leq \alpha$.
\end{enumerate}
By convention, we let $v_k$  represent the $k^{th}$ element of vector $v$, the maximum of an empty set is $-\infty$ and the minimum of an empty set is $\infty$.
\end{lemma}
The second part of Lemma \ref{lem:selectionRegion} is based on Proposition \ref{prop:monotone} below. \begin{proposition}(\cite{lee2016exact}, Lemma A.1)\label{prop:monotone}Let $F_{\theta}(x) := F^{[a, b]}_{\theta, \sigma^2}(x)$ denote the cumulative distribution function of a Gaussian random variable with mean $\theta$ and variance $\sigma^2$ whose domain of $x$ is $[a, b]$. $F_{\theta}(x)$ is monotone decreasing in $\theta$.
\end{proposition}

\bibliographystyle{agsm}
\bibliography{proximity.bib}
\end{document}